\documentclass[
  aps,
  prd,
  reprint,
  superscriptaddress,
  nofootinbib,
  floatfix
]{revtex4-2}

\usepackage[utf8]{inputenc}
\usepackage[english]{babel}

\usepackage{amsmath}
\usepackage{amsfonts}
\usepackage{amssymb}
\usepackage{amsthm}
\usepackage{mathtools}
\usepackage{graphicx}
\usepackage{booktabs}
\usepackage{array}
\usepackage{multirow}
\usepackage{xcolor}
\usepackage{listings}
\usepackage{algpseudocode}

\usepackage[
  colorlinks=true,
  allcolors=blue,
  pdfauthor={V. Pierro, V. Fiumara, V. Granata, G. Avallone},
  pdftitle={Thermal Noise in Optical Coatings: A Rigorous Opto-Mechanical Comparison of Nanolaminates and Isotropic Mixtures}
]{hyperref}

\definecolor{loopcolor}{RGB}{0, 100, 80}
\definecolor{codegray}{rgb}{0.5,0.5,0.5}
\definecolor{codepurple}{rgb}{0.58,0,0.82}
\definecolor{backcolour}{rgb}{0.95,0.95,0.92}

\algrenewcommand\algorithmicfor{\textcolor{loopcolor}{\textbf{for}}}
\algrenewcommand\algorithmicdo{\textcolor{loopcolor}{\textbf{do}}}
\algrenewcommand\algorithmicend{\textcolor{loopcolor}{\textbf{end}}}
\algrenewtext{EndFor}{\textcolor{loopcolor}{\textbf{end for}}}

\newtheorem{theorem}{Theorem}[section]
\newtheorem{lemma}[theorem]{Lemma}

\newtheorem{proposition}[theorem]{Proposition}

\theoremstyle{definition}

\theoremstyle{remark}
\newtheorem{remark}[theorem]{Remark}

\newcommand{\Imag}[1]{\operatorname{Im}\!\left[#1\right]}
\newcommand{\Real}[1]{\operatorname{Re}\!\left[#1\right]}

\begin{document}

\title{Thermal Noise in Optical Coatings: A Rigorous Opto-Mechanical Comparison of Nanolaminates and Isotropic Mixtures}

\author{Vincenzo Pierro}
\affiliation{Department of Engineering DING, University of Sannio, I-82100 Benevento, Italy}
\affiliation{INFN Sezione di Napoli, Gruppo Collegato di Salerno, I-84084 Fisciano (SA), Italy}

\author{Vincenzo Fiumara}
\affiliation{Department of Engineering, University of Basilicata, I-85100 Potenza, Italy}
\affiliation{INFN Sezione di Napoli, Gruppo Collegato di Salerno, I-84084 Fisciano (SA), Italy}

\author{Veronica Granata}
\affiliation{Dipartimento di Ingegneria Industriale, Elettronica e Meccanica (DIIEM), Universit\'{a} degli Studi di Roma Tre, I-00146 Roma, Italy}
\affiliation{INFN Sezione di Napoli, Gruppo Collegato di Salerno, I-84084 Fisciano (SA), Italy}

\author{Guerino Avallone}
\email[Corresponding author: ]{guerino.avallone@sa.infn.it}
\affiliation{Dipartimento di Fisica "E. R. Caianello", Universit\'{a} di Salerno, I-84084 Fisciano (SA), Italy}
\affiliation{INFN Sezione di Napoli, Gruppo Collegato di Salerno, I-84084 Fisciano (SA), Italy}


\begin{abstract}
We investigate whether a perfectly ordered periodic nanolaminate is intrinsically less mechanically 
dissipative than an isotropic co-sputtered mixture designed to provide the exact same optical function. 
Focusing on a benchmark room-temperature $\mathrm{SiO_2/TiO_2}$ material system, we compare coating Brownian
dissipation under strict parity of optical thickness and effective refractive index. The constituent amorphous phases are modeled using
complex Young's moduli and purely real Poisson's ratios, corresponding to a single-loss-angle approximation within each phase.
A priori, the outcome of this comparison is highly non-trivial due to competing topological effects: while the nanolaminate is optically optimal, its structural anisotropy forces in-plane 
shear deformations to operate at the highly dissipative Voigt upper bound, mechanically penalizing the structure compared to a well-mixed isotropic medium. 
The nanolaminate is treated deterministically through the exact 
long-wavelength optical (Wiener) and mechanical (Backus) effective-medium relations. 
Conversely, the isotropic mixture is described by realistic optical closures (Lorentz--Lorenz and Bruggeman) 
coupled to a complex mechanical self-consistent scheme (SCS). To assess the residual uncertainty associated with the unknown co-sputtered morphology, 
we explore the rigorous optical and mechanical admissible regions in the complex plane. 
We demonstrate that, owing to the small intrinsic optical and mechanical losses of the $\mathrm{SiO_2/TiO_2}$ pair, 
these rigorous admissible regions collapse into extremely narrow lenses, drastically restricting the microstructural uncertainty.
Moreover, across all realistic isotropic closures, the nanolaminate systematically 
exhibits lower mechanical dissipation than the corresponding co-sputtered mixture. 
Crucially, even when evaluated against the absolute thermodynamic lower bound of the isotropic state, this topological advantage remains strictly valid in the high effective refractive index regime.
We prove that the ordered architecture exploits the optical upper bound to drastically minimize 
the required volume fraction of the highly dissipative high-index phase, an effect that overwhelmingly overcompensates for its intrinsic mechanical anisotropy penalty.
\end{abstract}

\maketitle

\section{Introduction}
\label{sec:intro}

The mechanical dissipation of optical coatings remains a central limitation in precision interferometry. 
In high-stability optical systems, such as ground-based gravitational-wave detectors (e.g. Virgo \cite{Virgosite, virgo}, LIGO \cite{LIGOsite, ligo}, KAGRA\cite{KAGRAsite, kagra}), 
coating Brownian noise directly dictates the fundamental displacement noise floor in the most sensitive frequency bands \cite{Harry2002, Flaminio2010}. Consequently, 
the search for improved coating materials is a highly active field of research in order to go a step further than the already 
significant results obtained since the first detection of gravitational waves \cite{Abbott2016}. 

While significant efforts have recently been directed toward exploring novel low-loss
 amorphous phases materials, such as titania doped silica (TiO$_2$-doped SiO$_2$ \cite{Harry2006, McGhee_TiO2SiO2, Fazio2025}), 
 titania doped germania (TiO$_2$-doped GeO$_2$ \cite{VajenteSupp2021}) and advanced ternary systems 
 \cite{Pierro2021, Pierro2025}, experimentally optimizing the intrinsic loss angles of the constituent materials is only part of the solution. 
 It is equally crucial to rigorously address the role of topology and microstructure in determining the final macroscopic dissipative response of the coating.

The quest for optimal coating architectures has driven a decade-long debate within the gravitational-wave community regarding the best approach 
to synthesize materials with tailored refractive indices and low mechanical dissipation. Historically, the standard approach has been the co-deposition
 (co-sputtering) of different oxides to form a macroscopically homogeneous, amorphous isotropic mixture. However, an alternative and highly
 promising paradigm has been pioneered and extensively investigated: the use of perfectly ordered periodic nanolaminates, or 1D 
 photonic superlattices \cite{Pinto_G1301061, Pierro_ScienceDirect2023, Pinto_G1600665, Pinto_G1902307}. 

Recent advancements in plasma-assisted electron-beam deposition \cite{Neilson2024} have proven that such ultra-thin, perfectly ordered superlattices can be reliably manufactured with sub-nanometer precision 
\cite{VGranata_G2401236}, yielding metamaterials that strictly preserve their structural periodicity and amorphous nature of the nanolayered constituent materials even after 
high-temperature annealing \cite{Durante2024}, as confirmed using structural and optical characterization techniques similar to those employed 
for $\mathrm{TiO_2{:}Ta_2O_5}$ coatings in Ref.~\cite{Durante_CQG2024}.

By alternating sub-wavelength layers of pure low-index and high-index materials, it is possible to synthesize an effective metamaterial that matches the optical 
properties of a co-sputtered mixture, while fundamentally altering its structural and mechanical topology. 
Over the last decade, several experimental and theoretical 
campaigns have demonstrated the unique advantages of Nanolayered Optical Films (NLOF). Crucially, embedding a high-index material (e.g., TiO$_2$) within ultra-thin layers buffered 
by a glass-former (e.g., SiO$_2$) geometrically frustrates crystallization. This allows the composite to withstand significantly higher annealing 
temperatures without developing scattering crystallites, 
which in turn drastically reduces the internal mechanical friction \cite{Durante2023, Pinto_G2400526, Durante_Surfin2023_vacancies, Durante_Nano2021}.
Furthermore, nanolayering has been shown to flatten and shift the low-temperature mechanical loss peaks (cryo-peaks) \cite{Pinto_G2001499}, 
making this architecture a prime candidate for third-generation cryogenic detectors like the Einstein Telescope \cite{ETsite} and Cosmic Explorer \cite{CEsite}.

While the experimental and technological advantages of nanolaminates are increasingly well documented, their theoretical mechanical modeling 
has historically relied on simplified scalar approximations. Early estimates of coating thermal noise often treated the multilayer stack as an equivalent isotropic medium,
attempting to describe its dissipation through a single effective loss angle $\phi_{eff}$ and a scalar Young's modulus \cite{Pinto_G1401358, Pinto_G1200976}. 
However, such scalar homogenizations fundamentally fail to capture the intrinsic mechanical anisotropy of the layered structure. 

Under the Gaussian pressure profile dictated by Levin's direct approach \cite{Levin1998},
 a nanolaminate experiences a complex 3D stress field where in-plane shear and out-of-plane compressions dissipate energy at vastly different rates. 
 Therefore, to definitively assess the thermal noise performance of these architectures, one must abandon scalar loss angles and directly evaluate the elastic energy 
dissipated within the exact anisotropic stiffness tensor.

Building upon this extensive experimental and phenomenological background, the purpose of the present theoretical work is to definitively compare these two 
limiting architectures -- the disordered isotropic mixture and the perfectly ordered nanolaminate -- for the specific SiO$_2$/TiO$_2$ system under strict parity of optical design.

More precisely, we ask whether a nanolaminate is intrinsically less dissipative than an isotropic co-sputtered mixture designed to exhibit
the exact same effective refractive index and optical thickness. 
We focus specifically on coating Brownian thermal noise because this internal friction mechanism
represents the dominant fundamental noise source limiting the astrophysical sensitivity of current ground-based gravitational-wave interferometers 
(see \cite{Bodya} for a comprehensive reference on the subject) in their most sensitive frequency band ($\sim 10$--$300\,\mathrm{Hz}$). 
The opto-mechanical comparison is performed using the specific benchmark room-temperature properties of the $\mathrm{SiO_2/TiO_2}$
pair listed in Table~\ref{tab:material_properties}, assuming ideal
sharp interfaces and homogeneous internal dissipation within the individual constituent amorphous phases. The constituent amorphous phases are modeled using
complex Young's moduli and strictly real Poisson's ratios, thereby
excluding a dissipative component of the Poisson's ratios.

This comparison is nontrivial because the ordered layered structure is naturally anisotropic at the effective-medium level, whereas the isotropic co-sputtered morphology is not uniquely defined 
and must be described through effective-medium closures and admissible microstructural bounds \cite{Pinto_G1401358, Pinto_G1200976}. The nanolaminate can be treated deterministically through 
the exact long-wavelength optical and mechanical effective-medium relations for layered media, namely the Wiener optical description \cite{Wiener1912, BornWolf1999} 
and the Backus mechanical averaging procedure \cite{Backus1962}. 

By contrast, the isotropic co-sputtered mixture requires a more indirect construction. 
The phase fraction is first inferred from realistic optical closures, such as the Lorentz--Lorenz (L--L) and 
Bruggeman-type models \cite{Bruggeman1935, Polder1946}. This fraction is then mapped into effective complex bulk and shear moduli through a nonlinear self-consistent mechanical 
scheme \cite{Hill1965, Budiansky1965, Barta1984}. 
To assess the residual uncertainty associated with the unknown isotropic microstructure, we further consider rigorous optical and mechanical admissible regions in the complex plane. 
For the optical response, we rely on the exact bounds for complex dielectric permittivities developed by Bergman and Milton \cite{Bergman1980, Bergman1982, Milton1980, Milton1981}. 
For the mechanical response, we employ the extension of the Hashin-Shtrikman variational principles to complex viscoelastic moduli formulated by Gibiansky, Milton, 
and Berryman \cite{Gibiansky1993, Milton1997, Gibiansky1999}.

A central result of the present analysis is that, for the low-loss SiO$_2$/TiO$_2$ system, the exact admissible optical and mechanical regions collapse numerically into extremely narrow lenses. As a consequence, the rigorous bounds become very close to simple reduced geometric envelopes: short boundary segments in the optical plane 
and a narrow four-point Hashin-Shtrikman-Walpole (HSW)-like quadrilateral in the mechanical plane. This observation is important because it explains why the simplified full phase-space scan developed in this work remains quantitatively very close to the exact admissible set, despite not being presented as a fully rigorous coupled optical--mechanical extremal construction for a single realizable microstructure.

The paper is organized as follows. Section \ref{sec:thermal_noise_theory} describes the thermal noise model based on Levin's approach \cite{Levin1998}.  
Section~\ref{sec:nanolaminate} establishes the ordered nanolaminate as the deterministic reference architecture and derives its effective optical 
and mechanical properties. Section~\ref{sec:isotropic_mixtures} introduces the realistic isotropic-mixture closures used to model co-sputtered amorphous films. 
Section~\ref{sec:bounds} discusses the rigorous admissible regions for the optical and mechanical effective properties and explains their low-loss reduction in the present system. Section~\ref{sec:framework} combines these ingredients into the computational framework used for the full phase-space scan. 
In Section~\ref{sec:results} we present the dissipation comparison between the two architectures and show that, within all realistic isotropic closures considered here, 
the nanolaminate remains less dissipative than the corresponding isotropic co-sputtered mixture, especially in the high-index region.
Section~\ref{sec:conclusions} summarizes the main findings and discusses the broader implications of this topological advantage
for the design of next-generation gravitational-wave detectors.
Finally, in Appendices~\ref{app:min_fH} and \ref{app:algebraic_certification} it is mathematically demonstrated that for any target refractive index
the nanolaminate requires a strictly smaller volume fraction of the highly dissipative high-index oxide (TiO$_2$) than any admissible isotropic morphology. 

\section{Thermal Noise Evaluation via Levin's Approach}
\label{sec:thermal_noise_theory}

The ultimate metric for comparing different coating architectures is their contribution to the Brownian thermal noise of the mirror. 
Among the pioneering contributors to the analysis of
thermal noise in precision optical systems are Levin and Braginsky
\cite{Levin1998,Braginsky1999,Braginsky2000}. In the present work, we
focus specifically on the coating Brownian contribution. According to
Levin's theorem, which provides an operational formulation derived from
the Fluctuation-Dissipation Theorem, the displacement spectral density
is directly related to the time-averaged mechanical power dissipated
when a virtual, harmonically oscillating pressure profile $p(r)$
(a \textit{gedanken} force field) is applied to the mirror surface.
We therefore compute the associated dissipated power $W_{\rm diss}$
under this virtual pressure.

For interferometric applications, the virtual pressure matches the intensity profile of the 
interrogating laser beam. While specific experimental metrology setups may employ non-Gaussian or mesa beam profiles \cite{Bodya}, 
the fundamental Gaussian mode provides the canonical baseline for ground-based gravitational-wave detectors.

In the present calculations, we adopt the standard Gaussian profile:
\begin{equation}
    p(r) = \frac{2F_0}{\pi w_0^2} \exp\left(-\frac{2 r^2}{w_0^2}\right),
\end{equation}
where $F_0$ is the virtual force amplitude and $w_0$ is the beam radius (i.e. the laser waist on the mirror). 

For a thin coating of thickness $d_{coat}$ deposited on a semi-infinite substrate, the complex strain energy density can be analytically integrated. As formalized in Refs. \cite{FejerT2100186, VajenteSupp2021}, the macroscopic mechanical response of any (transversely isotropic or isotropic) coating can be fully encapsulated by four effective parameters: $A$, $B$, $D$, and $R_{mix}$. 

By defining the elastic moduli of the constituent materials as complex dynamic quantities, 
the mechanical dissipation is rigorously captured by the imaginary parts of these effective parameters. 
Assuming a perfectly elastic and lossless semi-infinite substrate,
the total power dissipated within the coating is given by:
\begin{equation} \label{eq:Wdiss}
\begin{split}
    W_{diss} = \omega d_{coat} \bigg[ & \frac{1}{2}\Imag{A}\,I_{SD} + \frac{1}{2}\Imag{B}\,I_{\Delta} \\
    & - \frac{1}{2}\Imag{D}\,I_{T} - R_{mix}\,I_{SDT} \bigg] ,
\end{split}
\end{equation}
where $\omega= 2 \pi f$ is the mechanical angular frequency. The terms $I_{SD}$, $I_{\Delta}$, $I_{T}$, and $I_{SDT}$ are geometric integrals over the substrate deformations, defined by the Lam\'e coefficients of the substrate ($\lambda_s$, $\mu_s$):
\begin{equation}
\begin{aligned}
    I_{SD} &= I_{\Delta} = \frac{F_0^2}{8\pi(\lambda_s + \mu_s)^2 w_0^2}, \\
    I_T &= \frac{2F_0^2}{\pi w_0^2}, \qquad 
    I_{SDT} = \frac{F_0^2}{2\pi(\lambda_s + \mu_s)w_0^2} .
\end{aligned}
\end{equation}
Note that a minus sign precedes $\Imag{D}$ because $D$ acts as a compliance-like quantity, possessing an imaginary part of opposite sign compared to the stiffness-like coefficients $A$ and $B$.
Furthermore, unlike the complex moduli $A$, $B$, and $D$, the cross-coupling term $R_{mix}$ appears without the imaginary operator because it is inherently 
defined as a strictly dissipative quantity (specifically, the imaginary part of the ratio between off-diagonal and diagonal stiffness components, as detailed in subsequent sections).

The primary goal of the theoretical frameworks developed in the following sections is to accurately compute the parameters $A$, $B$, $D$, and $R_{mix}$ for both the perfectly ordered nanolaminate (Sec.~\ref{sec:nanolaminate}) and the disordered isotropic mixture (Sec.~\ref{sec:isotropic_mixtures}), thereby enabling a direct and rigorous evaluation of Eq.~(\ref{eq:Wdiss}).

\section{Effective Properties of Nanolaminate}
\label{sec:nanolaminate}

Before addressing the complex microgeometries of co-deposited isotropic mixtures, we establish the theoretical baseline for perfectly ordered, periodically 
layered structures, commonly referred to as nanolaminates. 
We assume the individual layer thicknesses $d_L$ and $d_H$ are deeply sub-wavelength ($d_{L,H} \ll \lambda_0$ where $\lambda_0$ is the vacuum wavelength of the laser), 
allowing for a rigorous homogenization, while the total physical thickness of the stack $d_{coat}$ is constrained by the macroscopic optical design. 
When the thickness of
 individual layers is significantly smaller than the wavelength of the probing field (optical or mechanical), the composite behaves macroscopically as a 
homogeneous, but anisotropic, effective medium. It should be emphasized that the anisotropy of the nanolaminate is purely structural: each constituent phase is locally isotropic, 
and the effective anisotropic response emerges solely from the periodic stacking geometry.
As schematically shown in Fig.~\ref{fig:nanolaminate_schematic}, the ordered coating is modeled as a periodic alternation of low and high index materials sublayers deposited on the substrate.

\begin{figure}[t]
    \centering
    \includegraphics[width=\columnwidth]{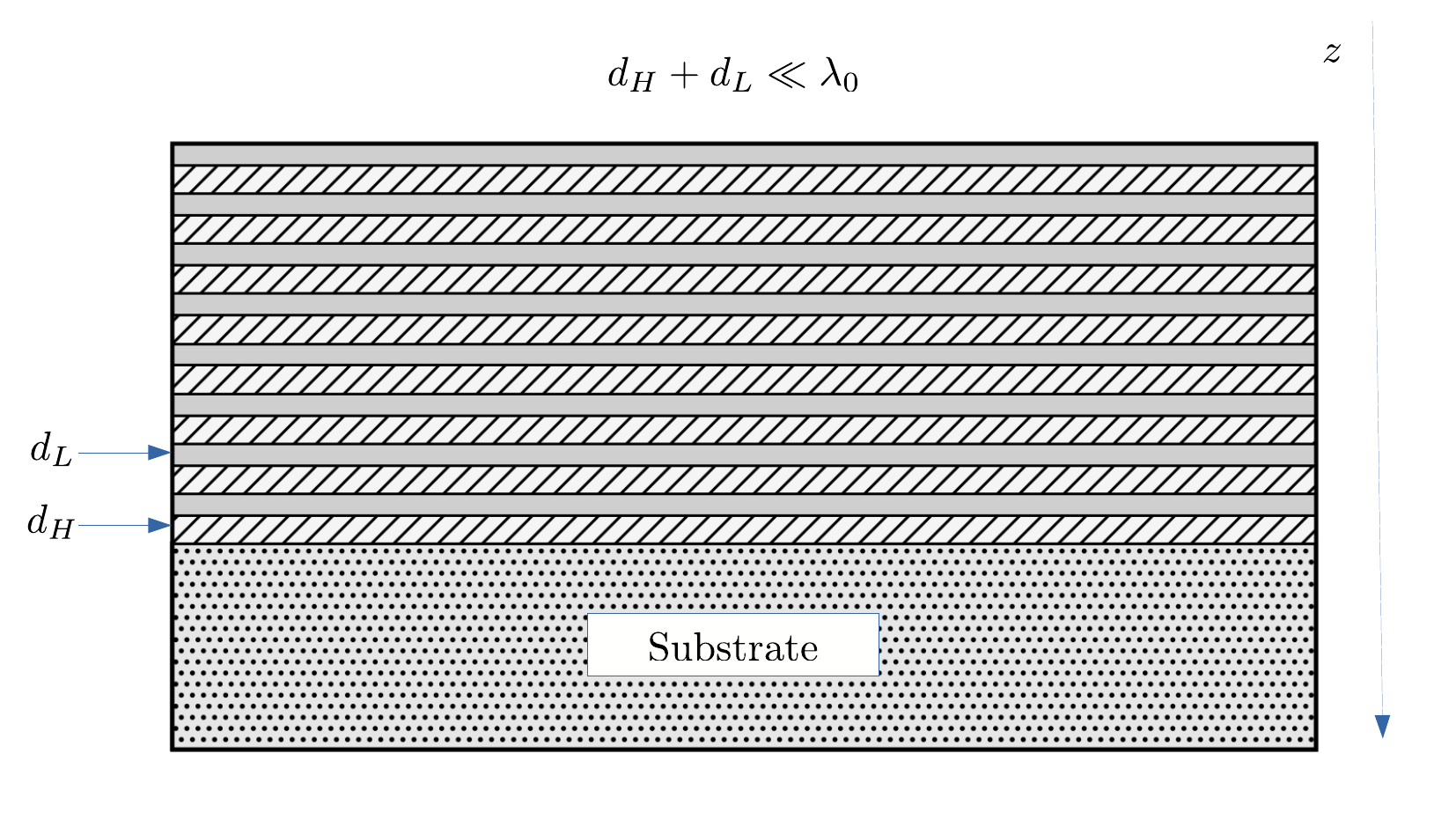}
    \caption{Schematic representation of the ordered periodic nanolaminate deposited on a fused-silica substrate. 
    The coating is modeled as an alternating sequence of low-index and high-index layers with period $d_H+ d_L \ll \lambda_0$, so that the structure can be 
    treated in the long-wavelength limit as an effective transversely isotropic medium.
    On the right-hand side of the figure, the $z$-axis is shown entering the laminated nano-structure and indicating the direction of the incident electromagnetic field.}
    \label{fig:nanolaminate_schematic}
\end{figure}

Recent morphological characterizations via Scanning Transmission Electron Microscopy and X-Ray Reflectivity 
have definitively proven that such ultra-thin periodic structures can be physically realized with extremely sharp interfaces and negligible 
interdiffusion \cite{Durante2024}, fully justifying the application of exact stratified effective-medium relations.

\subsection{Optical Response of Layered Media}

Optically, a nanolaminate composed of two isotropic phases behaves as a uniaxial crystal. The effective complex dielectric permittivity exhibits two distinct principal components depending on the polarization of the electric field relative to the layer interfaces. The classical optical limits for such stratified structures are defined by the Wiener bounds \cite{Wiener1912}, extensively detailed in standard electromagnetic theory for layered media \cite{sihvola2016,BornWolf1999}.

When the electric field is polarized parallel to the layers (in-plane), the effective permittivity $\varepsilon_{\parallel}$ is governed by the arithmetic mean:
\begin{equation}
    \varepsilon_{\parallel} = f_H \varepsilon_H + f_L \varepsilon_L
    \label{eq:wigbound}
\end{equation}
where $f_H = 1 - f_L$ is the volumetric fraction of high refractive index medium (H). Conversely, when the electric field is polarized perpendicular to the layers (out-of-plane), the effective permittivity $\varepsilon_{\perp}$ follows the harmonic mean (Reuss limit):
\begin{equation}
    \varepsilon_{\perp} = \left( \frac{f_H}{\varepsilon_H} + \frac{f_L}{\varepsilon_L} \right)^{-1}.
    \label{eq:reussbound}
\end{equation}
The exact analytical expressions (\ref{eq:wigbound}) (\ref{eq:reussbound}) represent the extreme bounds of structural optical anisotropy, providing a deterministic reference 
against which the optical response of more disordered morphologies can be evaluated.
 Under standard (normal incidence) reflectometric measurements of thin films deposited on flat substrates---reproducing the normal incidence of laser light on the coating surface of end masses in interferometric gravitational-wave detectors---the primary optical response is determined by the in-plane refractive index, $\tilde{n}_{\parallel} = \sqrt{\varepsilon_{\parallel}}$. 
Therefore, this effective index dictates the physical thickness $d_{\text{coat}}$ required to achieve a specific optical function, such as a high-reflectivity Bragg mirror.

For a target operating vacuum wavelength 
$\lambda_0$ (e.g., $\lambda_0 = 1064\text{ nm}$ for current gravitational-wave interferometers), the quarter-wave optical thickness 
condition is strictly defined as $d_{coat} = \lambda_0 / (4 \Real{\tilde{n}_{\parallel}})$.

\subsection{Mechanical Backus Averaging and Thermal Noise Parameters}

Symmetrically to the optical domain, the mechanical response of a nanolaminate is highly direction-dependent. A stratified medium composed of two isotropic materials (with stiffness components $c_{11}^{(i)}$, $c_{12}^{(i)}$, and $c_{44}^{(i)}$ for phases $i \in \{L,H\}$) macroscopically behaves as a transversely isotropic material with a vertical axis of symmetry.

The exact static effective stiffness tensor $\mathbf{C}^{NL}$ for this layered architecture was rigorously derived by Backus \cite{Backus1962}. Utilizing the volumetric averaging operator $\langle x \rangle = f_H x_H + f_L x_L$, the effective elastic constants governing the nanolaminate are algebraically defined as follows.

The out-of-plane compressional stiffness ($C_{33}$) is strictly governed by the harmonic average of the individual longitudinal stiffnesses:
\begin{equation}
    C_{33} = \langle c_{11}^{-1} \rangle^{-1} .
\end{equation}

The coupling between out-of-plane and in-plane deformations is given by:
\begin{equation}
    C_{13} = C_{33} \left\langle \frac{c_{12}}{c_{11}} \right\rangle .
\end{equation}

The in-plane compressional and transverse stiffnesses ($C_{11}$ and $C_{12}$) are mathematically more coupled, combining arithmetic averages of the plane-strain moduli with a correction proportional to the out-of-plane coupling:
\begin{equation}
    C_{11} = \left\langle c_{11} - \frac{c_{12}^2}{c_{11}} \right\rangle + C_{33} \left\langle \frac{c_{12}}{c_{11}} \right\rangle^2
\end{equation}
\begin{equation}
    C_{12} = \left\langle c_{12} - \frac{c_{12}^2}{c_{11}} \right\rangle + C_{33} \left\langle \frac{c_{12}}{c_{11}} \right\rangle^2 .
\end{equation}

Finally, the in-plane shear modulus ($C_{66}$) obeys the Voigt limit \cite{Voigt1889} (arithmetic mean):
\begin{equation}
    C_{66} = \langle c_{44} \rangle .
\end{equation}

To evaluate the dissipated power using Eq.~(\ref{eq:Wdiss}), we map these tensorial components into the specific effective-medium parameters $A$, $B$, and $D$ formalized for interferometric mirrors \cite{FejerT2100186}:
\begin{equation}
\begin{aligned}
    A_{NL} &= \frac{1}{2} \left( C_{11} + C_{12} - \frac{2 C_{13}^2}{C_{33}} \right), \\
    B_{NL} &= C_{66}, \qquad D_{NL} = \frac{1}{C_{33}} .
\end{aligned}
\end{equation}
Additionally, the cross-coupling term $R_{mix}$ is defined as the imaginary part of the ratio between the off-diagonal and diagonal out-of-plane stiffness components:
\begin{equation}
    R_{mix}^{NL} = \Imag{\frac{C_{13}}{C_{33}}} .
    \label{eq:Rmix_NL}
\end{equation}
Under the assumptions adopted in this work -- namely, complex scalar Young's moduli and purely real Poisson's ratios for the constituent 
phases -- this coupling coefficient vanishes identically for the nanolaminate architecture ($R_{mix}^{NL} = 0$) \cite{FejerT2100186}.
These mapped parameters establish the absolute theoretical limits of structural anisotropy for the composite. 
By evaluating the complex energy density using these Backus-derived coefficients, the nanolaminate architecture provides a 
deterministic minimum-dissipation baseline against which the behavior of disordered isotropic mixtures can be quantitatively compared.

Crucially, to account for the intrinsic mechanical dissipation of the constituent materials, we invoke the elastic-viscoelastic correspondence principle. 
The constituent stiffnesses are defined as complex dynamic moduli, $c_{ij}^* = c_{ij}' + i c_{ij}'' = c_{ij}'(1 + i\phi_{ij})$, where $\phi_{ij}$ represents the specific mechanical loss angle of the phase.

Because the correspondence principle preserves the algebraic structure of the boundary conditions, the Backus averaging operators and the subsequent mapping into the $A_{NL}$, $B_{NL}$, and $D_{NL}$ parameters are strictly valid in the complex plane. In our computational framework, these effective parameters are evaluated using full complex arithmetic. This avoids the truncation errors inherent in first-order Taylor expansions and allows the effective mechanical dissipation of the composite to be determined directly from the imaginary components of the final macroscopic parameters.

\subsubsection{Validity of the Quasi-Static Approximation}

The exactness of the Backus averaging method relies on the condition $k_{sound} l' \ll 1$, where $k_{sound}$ is the acoustic wave number 
and $l'$ is the characteristic thickness of the layers. In our case, the frequency of interest for macroscopic mechanical perturbations is on the order of $f = 100\,\mathrm{Hz}$.

Given the macroscopic real part of Young's moduli of $\mathrm{SiO}_2$ and $\mathrm{TiO}_2$ ($72\,\mathrm{GPa}$ and $165\,\mathrm{GPa}$, respectively) and their standard densities, 
the acoustic wave velocity $v$ in both phases exceeds $5000\,\mathrm{m/s}$. 
Consequently, the acoustic wavelength at $100\,\mathrm{Hz}$ is extremely large ($\lambda_{sound} = v_{sound}/f \approx 50\,\mathrm{m}$).

For a typical nanolaminate with layer thicknesses $l'$ in the nanometer regime ($10^{-9}$--$10^{-8}\,\mathrm{m}$), the dimensionless parameter $k_{sound} l'$ is on the order of $10^{-8}$.
Therefore, the elastic field variation across individual layers is negligible, placing the analysis deeply within the quasi-static regime.
In this long-wavelength limit, the structural anisotropy is governed by the zero-order effective static moduli, and the Backus limits provide an exact analytical 
representation of the composite mechanics without dynamic scattering corrections.

\section{Effective Medium Approximations for Realistic Isotropic Mixtures}
\label{sec:isotropic_mixtures}

The characterization of co-deposited amorphous thin films requires the estimation of effective properties from constituent phases. 
In experimental scenarios, where microstructural details are often inaccessible, effective-medium theories (EMTs) serve as the standard predictive framework \cite{sihvola2016}. These models relate the macroscopic response of the composite to the volume fractions of its constituents by assuming a specific idealized topology. This section outlines the analytical models employed to establish a realistic baseline for the optical and mechanical properties of the isotropic mixture.

\subsection{Optical Effective-Medium Models}
\label{subsec:Optical}

In typical experimental scenarios, the effective real refractive index of the composite film ($n_{eff}$) is a known parameter, routinely measured via ellipsometry. 
To mathematically link the optical response to the mechanical one, we invert the classical optical EMTs to extract the effective volume fraction of the inclusions ($f_H$).

The first model considered is the L--L relation, mathematically equivalent to the Maxwell--Garnett approximation for vacuum-embedded inclusions. 
For a two-phase mixture with permittivities $\varepsilon_H$, $\varepsilon_L$, and effective permittivity $\varepsilon_{eff} = \tilde{n}_{eff}^2$, the L--L equation is expressed as
\begin{equation}
    \frac{\varepsilon_{eff} - 1}{\varepsilon_{eff} + 2} = f_L \frac{\varepsilon_L - 1}{\varepsilon_L + 2} + f_H \frac{\varepsilon_H - 1}{\varepsilon_H + 2} ,
\end{equation}
where $f_H = 1 - f_L$ is the volumetric fraction of phase 2.

To account for a fully symmetric topology where neither phase is distinctly the host matrix, we also use the \textbf{Bruggeman symmetric effective-medium theory} \cite{Bruggeman1935}. For spherical inclusions, the standard Bruggeman equation imposes
\begin{equation}
    f_L \frac{\varepsilon_L - \varepsilon_{eff}}{\varepsilon_L + 2\varepsilon_{eff}} + f_H \frac{\varepsilon_H - \varepsilon_{eff}}{\varepsilon_H + 2\varepsilon_{eff}} = 0 .
\end{equation}
Furthermore, to account for potential morphological anisotropies at the nanoscale, we implement a \textbf{generalized Bruggeman model} incorporating a depolarization (shape) factor $L$ \cite{Polder1946}:
\begin{equation}
    f_L \frac{\varepsilon_L - \varepsilon_{eff}}{\varepsilon_{eff} + L(\varepsilon_L - \varepsilon_{eff})} + f_H \frac{\varepsilon_H - \varepsilon_{eff}}{\varepsilon_{eff} + L(\varepsilon_H - \varepsilon_{eff})} = 0 .
\end{equation}
The case $L = 1/3$ recovers the spherical case, $L < 1/3$ represents more prolate inclusions, and $L > 1/3$ represents more oblate inclusions. 
By numerically inverting these equations, we compute the specific volume fractions corresponding to a given target real refractive index.
As rigorously proven in Appendix~\ref{app:min_fH}, for any target refractive index, the high-index volume fraction required by the nanolaminate is 
strictly smaller than that of any admissible isotropic morphology.

\subsection{Mechanical Self-Consistent Scheme}

Consistent with the treatment of the nanolaminate architecture, the constituent properties are introduced into the SCS as complex dynamic moduli. 
Following Barta's approach \cite{Barta1984}, this highly nonlinear coupled system is numerically solved in our framework using extended arithmetic precision in the complex plane. 
By avoiding linear approximations for the loss angles, the roots of these equations inherently yield both the real elastic constants (storage moduli) 
and the associated viscoelastic losses (loss moduli) of the effective isotropic medium.

Once the volume fraction is determined optically (as shown in Section \ref{subsec:Optical}), 
the effective mechanical properties (complex bulk modulus $K_{eff}$ and shear modulus $\mu_{eff}$) of the isotropic mixture are evaluated using the 
SCS. Originally established independently by Hill \cite{Hill1965} and Budiansky \cite{Budiansky1965}, the practical numerical formulation used here follows Barta \cite{Barta1984}. 
Unlike simple mixing rules, the SCS implicitly couples the bulk and shear responses, making it suitable for densely packed, strongly interacting amorphous mixtures. 
The effective moduli are found by simultaneously solving
\begin{equation}
    f_L \frac{K_L - K_{eff}}{K_{eff} + \alpha(K_L - K_{eff})} + f_H \frac{K_H - K_{eff}}{K_{eff} + \alpha(K_H - K_{eff})} = 0
\end{equation}
\begin{equation}
    f_L \frac{\mu_L - \mu_{eff}}{\mu_{eff} + \beta(\mu_L - \mu_{eff})} + f_H \frac{\mu_H - \mu_{eff}}{\mu_{eff} + \beta(\mu_H - \mu_{eff})} = 0 ,
\end{equation}
where the coupling coefficients $\alpha$ and $\beta$ depend on the effective properties themselves:
\begin{equation}
    \alpha = \frac{3K_{eff}}{3K_{eff} + 4\mu_{eff}}, \qquad \beta = \frac{6(K_{eff} + 2\mu_{eff})}{5(3K_{eff} + 4\mu_{eff})} .
\end{equation}
The system is solved directly in the complex plane, yielding both the real elastic constants and the associated viscoelastic losses of the effective isotropic medium.

Here, the subscripts $L$ and $H$ denote the low-index and high-index constituent phases, respectively. 
Their isotropic elastic properties are introduced through the complex bulk and shear moduli
$K_L$, $K_H$, $\mu_L$, and $\mu_H$, obtained from the corresponding Young's moduli
$Y_L^*$, $Y_H^*$ and Poisson ratios $\nu_L$, $\nu_H$ according to
\begin{equation}
    K_i = \frac{Y_i^*}{3(1-2\nu_i)},
    \qquad
    \mu_i = \frac{Y_i^*}{2(1+\nu_i)},
    \qquad i \in \{L,H\}.
\end{equation}
In the present work, the Young's moduli are treated as complex quantities,
$Y_i^* = Y_i(1+i\phi_i)$, $\phi_i$ being the mechanical loss angle, while the Poisson ratios are taken to be real.

A significant challenge in applying multi-parameter
viscoelastic formulations lies in the experimental determination of
the bulk ($\phi_{K,i}$) and shear ($\phi_{\mu,i}$) loss angles
separately. Many previous coating Brownian-noise models have adopted a
single loss angle for each constituent,
$\phi_{K,i}=\phi_{\mu,i}=\phi_i$, as a simplifying baseline. Generalized
viscoelastic formulations, however, allow the two dissipation channels
to be treated independently, and there is no first-principles reason
requiring their equality \cite{FejerT2100186,Vajente2020PRD}.
Experimental studies aimed at separating bulk and shear losses provide
evidence that the two channels may differ, but they have not yet
produced a mutually consistent and transferable set of loss parameters.
In particular, measurements on titania-doped tantala coatings
\cite{Vajente2020PRD,Abernathy2018PLA} obtained different results for
the relative magnitudes and frequency dependences of the bulk and shear
losses. Moreover, no constituent-specific set of both
$\phi_{K,i}$ and $\phi_{\mu,i}$ is available for the amorphous
$\mathrm{SiO_2}$ and $\mathrm{TiO_2}$ materials considered here.
We therefore adopt purely real constituent Poisson's ratios as a
controlled single-loss-angle baseline in the absence of consistent
experimental inputs.

From the resulting effective isotropic moduli $K_{eff}$ and $\mu_{eff}$, the stiffness tensor components are simply 
$C_{11}^{ISO} = K_{eff} + \frac{4}{3}\mu_{eff}$, $C_{12}^{ISO} = K_{eff} - \frac{2}{3}\mu_{eff}$, and $C_{66}^{ISO} = \mu_{eff}$. 

Consequently, the parameters required for the thermal noise evaluation Eq.~(\ref{eq:Wdiss}) are mapped as follows:
\begin{equation}
\begin{aligned}
    A_{ISO} &= \frac{1}{2}\left(C_{11}^{ISO}+C_{12}^{ISO}-2\frac{(C_{12}^{ISO})^2}{C_{11}^{ISO}}\right), \\
    B_{ISO} &= C_{66}^{ISO}, \qquad D_{ISO} = \frac{1}{C_{11}^{ISO}} .
\end{aligned}
\end{equation}
Unlike the nanolaminate case, the cross-coupling coefficient does not vanish for an isotropic mixture and is explicitly retained as a dissipative quantity:
\begin{equation}
    R_{mix}^{ISO} = \Imag{\frac{C_{12}^{ISO}}{C_{11}^{ISO}}} .
\end{equation}

\begin{table*}[t]
\centering
\begin{tabular}{@{}lccccc@{}}
\toprule
\multirow{2}{*}{\textbf{Material}} & \multicolumn{3}{c}{\textbf{Mechanical Properties}} & \multicolumn{2}{c}{\textbf{Optical Properties}} \\ \cmidrule(lr){2-4} \cmidrule(l){5-6}
 & $Y$ [GPa] & $\phi$ & $\nu$ & $n$ & $\kappa$ \\ \midrule
Silica (SiO$_2$)  & 72.0  & $0.5 \times 10^{-4}$ & 0.17 & 1.45 & $1.0 \times 10^{-8}$ \\
Titania (TiO$_2$) & 165.0 & $1.4 \times 10^{-4}$ & 0.28 & 2.47 & $1.0 \times 10^{-3}$ \\ \bottomrule
\end{tabular}
\caption{Mechanical and optical parameters of the constituent materials used in the theoretical bounding models.
The silica parameters constitute representative benchmark values for high-quality sputtered amorphous silica
\cite{Flaminio2010,Bodya}. The optical parameters assigned to the amorphous $\mathrm{TiO_2}$ constituent were
obtained from a refined analysis of the spectroscopic-ellipsometry data acquired during our experimental campaign on
$\mathrm{SiO_2/TiO_2}$ nanolayer samples \cite{Durante2024,Magnozzi2018}. The mechanical parameters assigned to
$\mathrm{TiO_2}$ are nominal room-temperature baseline inputs,
consistent with the broad ranges considered for amorphous high-index
coating materials in Ref.~\cite{Bodya}. The values listed in this table
should therefore not be interpreted as a parameter set measured jointly
on a single specimen.
The complex Young's modulus is defined as $Y^*_i = Y_i(1 + i\phi_i)$, where $Y_i$ is the storage modulus and $\phi_i$ 
is the mechanical loss angle of the $i$-th material. The Poisson ratios $\nu_i$ are taken to be real.
The complex refractive index is defined as $\tilde{n} = n + i\kappa$, where $n$ is the real refractive index and $\kappa$ is the extinction coefficient.}
\label{tab:material_properties}
\end{table*}

\section{Rigorous Bounds and Low-Loss Reduction of the Admissible Regions}
\label{sec:bounds}

The rigorous optical and mechanical bounds introduced in this section are not used merely as abstract
extremal results. Their main role is to quantify how much microstructural uncertainty remains for the
isotropic co-sputtered mixture once the constituent properties in Table~\ref{tab:material_properties}
are fixed. 

\subsection{Theoretical Bounds for the Complex Refractive Index}

To accurately model the optical response of the two-phase composite mixture (e.g., SiO$_2$ and TiO$_2$), it is
necessary to determine its effective complex refractive index, $\tilde{n}_{eff} = n_{eff} + i \kappa_{eff}$, where
$n_{eff}$ is the real refractive index and $\kappa_{eff}$ is the extinction coefficient. The complex refractive
index is fundamentally related to the effective complex dielectric permittivity via
$\varepsilon_{eff} = \tilde{n}_{eff}^2$.

Since the exact microgeometry of the co-deposited mixture is generally unknown, classical effective-medium
approximations may yield inaccurate predictions, especially for optical losses. To overcome this limitation,
we employ the rigorous bounding theory for the complex dielectric constant of two-component composites
developed by Bergman \cite{Bergman1980,Bergman1982} and Milton \cite{Milton1980,Milton1981,Milton2002}.

For a macroscopically homogeneous and isotropic three-dimensional composite with known volume fractions $(f_H, f_L=1-f_H)$ 
and known constituent complex permittivities $(\varepsilon_L, \varepsilon_H)$, the effective permittivity $\varepsilon_{eff}$ is strictly confined within a closed, 
lens-shaped region $\mathcal{L}$ in the complex $\varepsilon$ plane. 
The exact boundary of this region, formed by two circular arcs corresponding to the topological extremes where either 
the low-index or the high-index material acts as the host phase, can be explicitly parameterized by two real variables $u$ and $v$ 
(related to the geometric depolarization factors of the inclusions) as follows:
\begin{equation}
\label{eq:bergman_lens}
\left\{
\begin{aligned}
   \varepsilon_{eff}^{LH} &= f_L\varepsilon_L + f_H\varepsilon_H - \frac{f_L f_H (\varepsilon_L - \varepsilon_H)^2}{3 \left[ u\varepsilon_L + (1 - u)\varepsilon_H \right]}, \\
   &\quad\text{for } \tfrac{f_H}{3} \le u \le 1 - \tfrac{f_L}{3}, \\[8pt]
   \varepsilon_{eff}^{HL} &= f_L\varepsilon_L + f_H\varepsilon_H - \frac{f_L f_H (\varepsilon_L - \varepsilon_H)^2}{3 \left[ v\varepsilon_H + (1 - v)\varepsilon_L \right]}, \\
   &\quad\text{for } \tfrac{f_L}{3} \le v \le 1 - \tfrac{f_H}{3}.
\end{aligned}
\right.
\end{equation}
Evaluating the first branch $\varepsilon_{eff}^{LH}$ traces the primary circular arc, while evaluating the second branch $\varepsilon_{eff}^{HL}$ traces the complementary arc. 
The two vertices where these arcs intersect correspond precisely to the exact analytical solutions for the \textit{Doubly Coated Sphere}
 microgeometries, representing the theoretical extremes of the composite's optical response.
A central topological feature, proven analytically in Appendix~\ref{app:min_fH} and certified algebraically in Appendix~\ref{app:algebraic_certification},
 is that this entire admissible region for isotropic media is strictly 
 bounded away from the anisotropic Wiener upper bound, establishing a fundamental optical performance gap between the two architectures.
The non-intersection between the Bergman--Milton boundary arcs and the Wiener limit is algebraically certified using the quantifier-elimination
 procedure detailed in Appendix~\ref{app:algebraic_certification}. 

\begin{figure*}[b]
    \centering
    \includegraphics[width=0.88\textwidth]{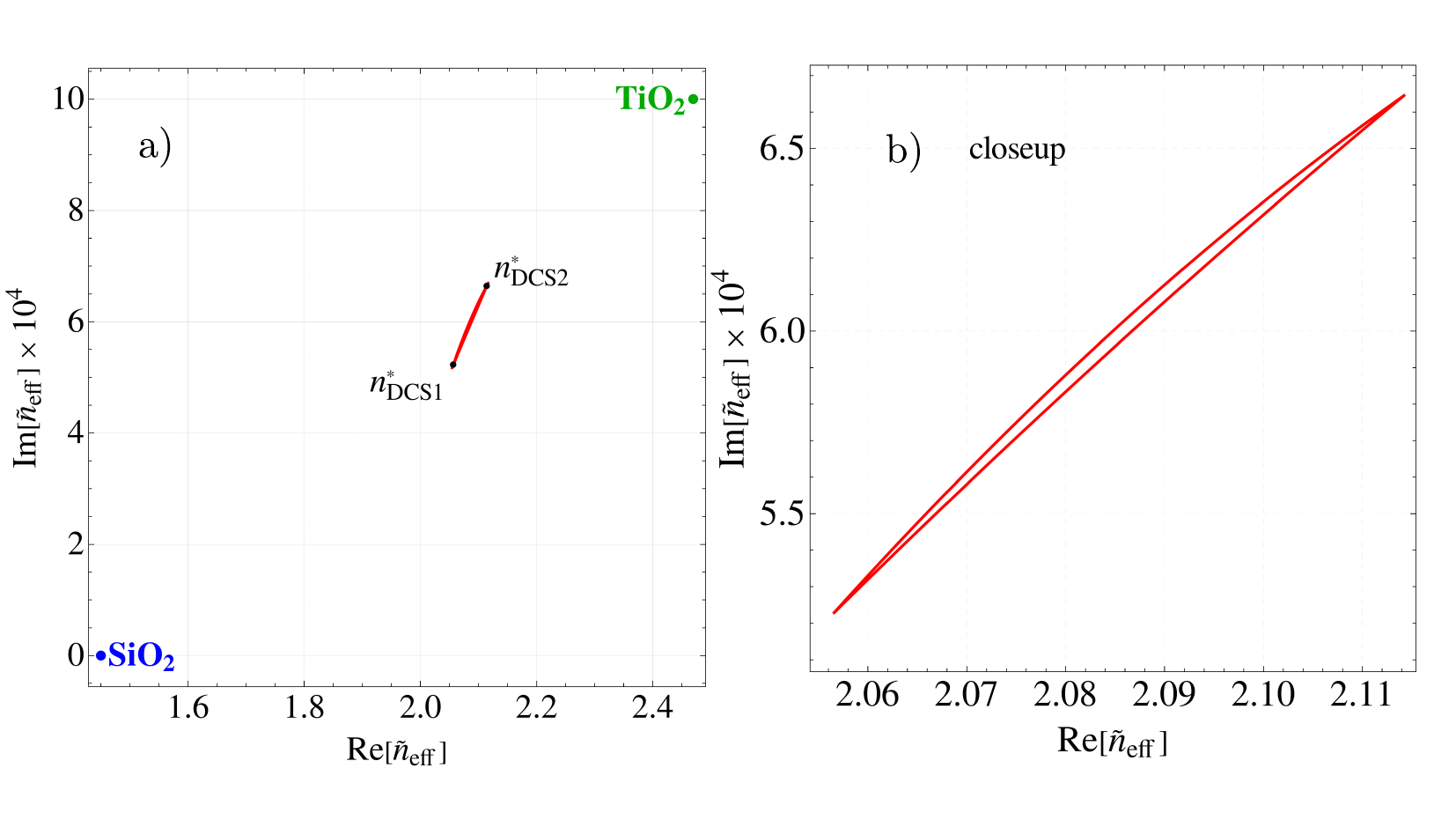}
    \caption{Rigorous bounds for the effective complex refractive index ($\tilde{n}_{eff}$) of the $\mathrm{SiO}_2/\mathrm{TiO}_2$ composite. 
    \textbf{(a)} Macroscopic view in the complex plane showing the constituent phases (Phase 1: blue label $\mathrm{SiO}_2$, Phase 2: green label $\mathrm{TiO}_2$) and the computed bounding region. 
    Note the application of a $10^4$ scaling factor on the imaginary axis to account for the intrinsically low optical losses. 
    \textbf{(b)} A high-resolution closeup of the mapped Bergman--Milton region. 
    The strictly permissible domain is confined within the red boundaries, which connect the two theoretical extremes corresponding to the doubly coated sphere geometries, 
    denoted as $n_{DCS1}^*$ and $n_{DCS2}^*$. 
    The extreme narrowness of this lens-shaped region demonstrates that the effective optical losses of the composite are tightly constrained, rendering them largely independent of the specific microscopic morphology.}
    \label{fig:refractive_bounds}
\end{figure*}

\begin{figure*}[b]
    \centering
    \includegraphics[width=0.75\textwidth]{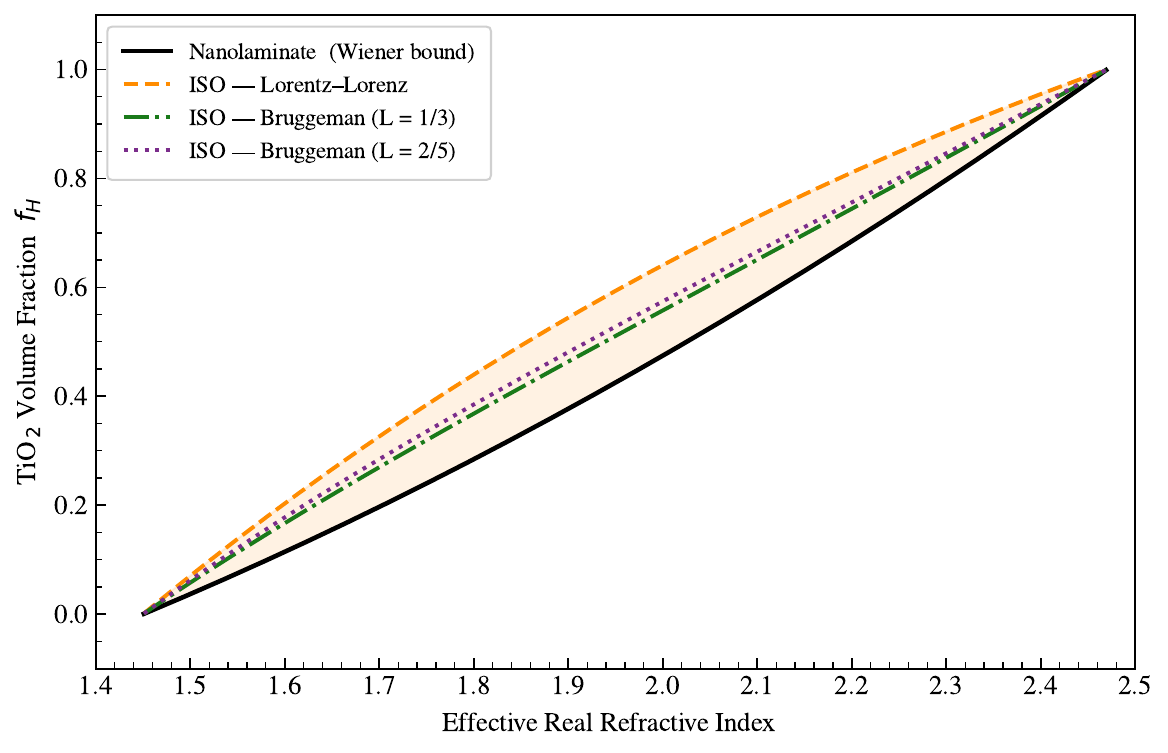}
    \caption{Exact high-index (TiO$_2$) volume fraction $f_H$ required to achieve a target real effective refractive index $n_{\rm eff}$.
     The curves are computed using the constituent properties from Table~\ref{tab:material_properties}. 
     The solid black line represents the perfectly ordered nanolaminate (Wiener upper bound), while the dashed and dotted lines correspond to 
     standard isotropic effective-medium closures. 
     The shaded area highlights the reduced volume fraction used by nanolaminated composite. 
     (see Appendix~\ref{app:min_fH} for the complete analytical proof of this systematic volume saving).}
     \label{fig:fH_comparison}
\end{figure*}

In our computational model, the exact boundary of the allowed region in the $\varepsilon$ plane is
calculated for a given volume fraction and subsequently mapped into the complex refractive-index plane
by taking the principal square root, $\tilde{n}=\sqrt{\varepsilon}$.

Figure~\ref{fig:refractive_bounds} shows that, for the present SiO$_2$/TiO$_2$ system, the optical
admissible region collapses into a very thin lens once the intrinsic extinction of the constituents is
taken into account. In practice, once the real part of the refractive index is fixed, the admissible
interval for the imaginary part becomes extremely narrow. This low-loss optical collapse is one of the
reasons why the final full scan can be simplified without losing the main physical content.

The structure of the allowed region for the dielectric constant of an isotropic mixture in the complex plane allows us to formulate 
the analytical (general) results derived in Appendix~\ref{app:min_fH} and the numerical evidence presented in Fig.~\ref{fig:fH_comparison}, 
which elucidate that the nanolaminate architecture strictly minimizes the high-index volume fraction compared to any isotropic closure,
 a feature that remains consistently valid within the narrow rigorous bounds of the admissible region.

\begin{figure*}[t]
    \centering
    \includegraphics[width=0.88\textwidth]{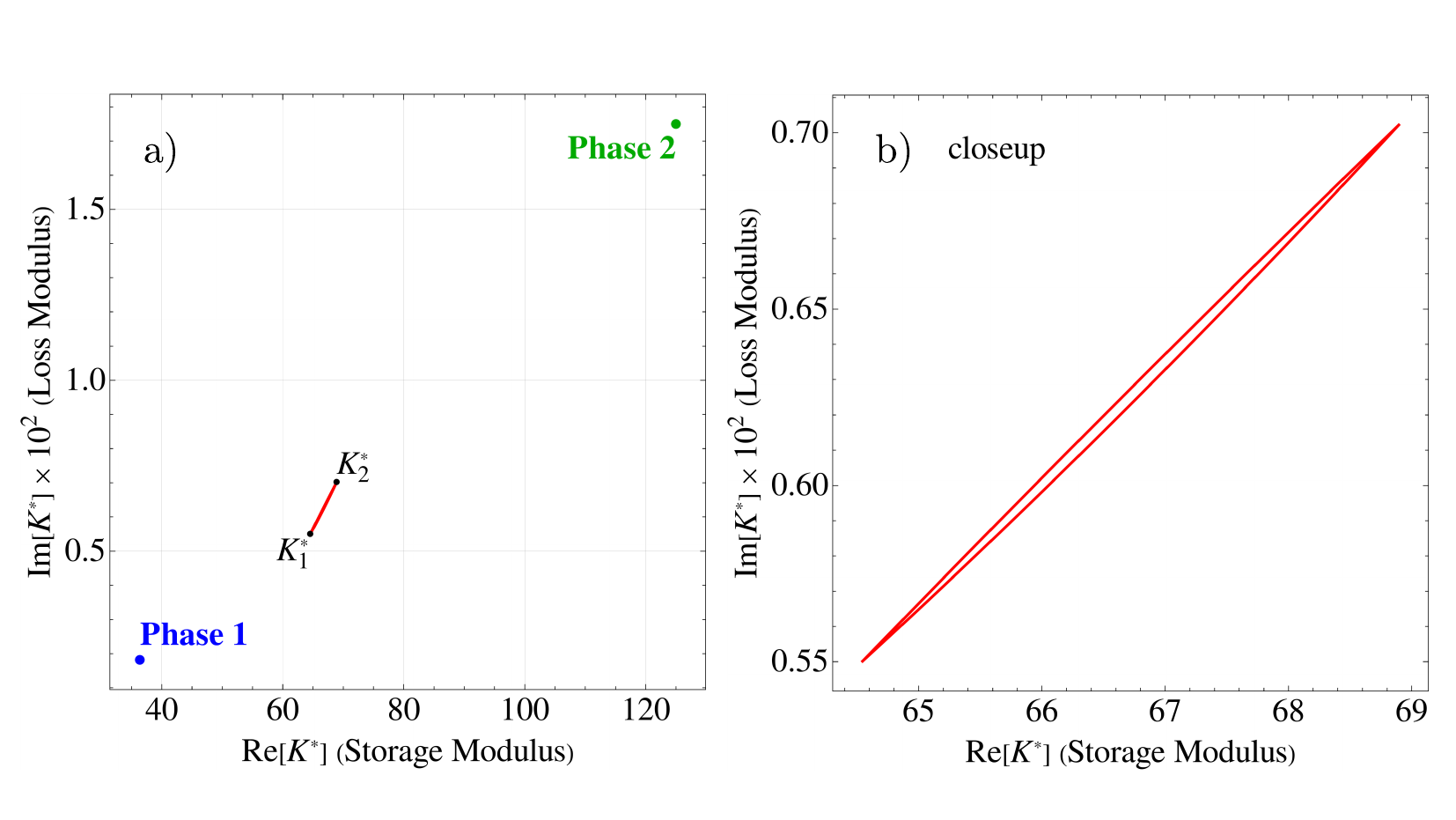}
    \caption{Rigorous bounds for the effective complex bulk modulus $K^*$ of the $\mathrm{SiO}_2/\mathrm{TiO}_2$ composite. 
    \textbf{(a)} Macroscopic view of the complex plane showing the properties of the pure constituent phases 
    (Phase 1: $\mathrm{SiO}_2$ blue dot; Phase 2: $\mathrm{TiO}_2$ green dot) and the bounding region (red) for a fixed volume fraction. The black dots, labeled $K_1^*$ and $K_2^*$, 
    represent the complex Hashin--Shtrikman--Walpole points. 
    \textbf{(b)} A high-resolution closeup of the bounding region. Due to the extremely low mechanical loss tangents of the constituent materials ($\phi \sim 10^{-4}$), 
    the bounding area degenerates into a highly elongated, eye-shaped lens. For practical modeling purposes, this microstructural uncertainty region is extremely 
    narrow and can be well approximated by a straight line segment connecting the two vertices.}
    \label{fig:bulk_bounds}
\end{figure*}

\subsection{Theoretical Bounds for the Complex Elastic Moduli}

The theoretical prediction of the mechanical response of a two-phase composite such as a
SiO$_2$/TiO$_2$ mixture requires the determination of its effective macroscopic properties.
While classical mixing rules provide simple arithmetic limits, they are inadequate for complex
viscoelastic parameters, where the moduli possess both a real (storage) and an imaginary (loss)
component. To accurately constrain the effective complex bulk ($K^*$) and shear ($\mu^*$)
moduli without prior knowledge of the exact microgeometry, we employ the rigorous variational
bounding method developed by Gibiansky, Milton, and Berryman
\cite{Gibiansky1993,Milton1997,Gibiansky1999}.

For the complex bulk modulus, the admissible region shown in Fig.~\ref{fig:bulk_bounds} is already
extremely narrow. The closeup of Fig.~\ref{fig:bulk_bounds} makes clear that the low constituent
loss angles listed in Table~\ref{tab:material_properties} compress the rigorous lens into a thin strip,
so that the exact admissible set is numerically very close to the straight segment joining the two HSW
endpoints \cite{Gibiansky1993,Hashin1963,Walpole1966}. In practice, this means that the uncertainty
associated with the isotropic co-sputtered morphology is already strongly reduced at the level of the
complex bulk response.

To avoid excessive notational clutter, in this subsection only we denote the low-index and high-index
phases by the indices $1$ and $2$, respectively.

For the complex shear modulus, the problem is mapped into an auxiliary complex plane using the
$Y$-transform \cite{Milton1997}. 
The effective shear modulus $\mu^*$ is related 
to a complex variable $y_\mu$ (the \mbox{$Y$-parameter}, which encapsulates the microstructural geometric constraints) through the fractional linear transformation
\begin{equation}
    \mu^* =
    \left[
    \sum_{i=1}^{2}\frac{f_i}{\mu_i+y_\mu}
    \right]^{-1}
    - y_\mu .
\end{equation}

\begin{figure}[t]
    \centering
    \includegraphics[width=\columnwidth]{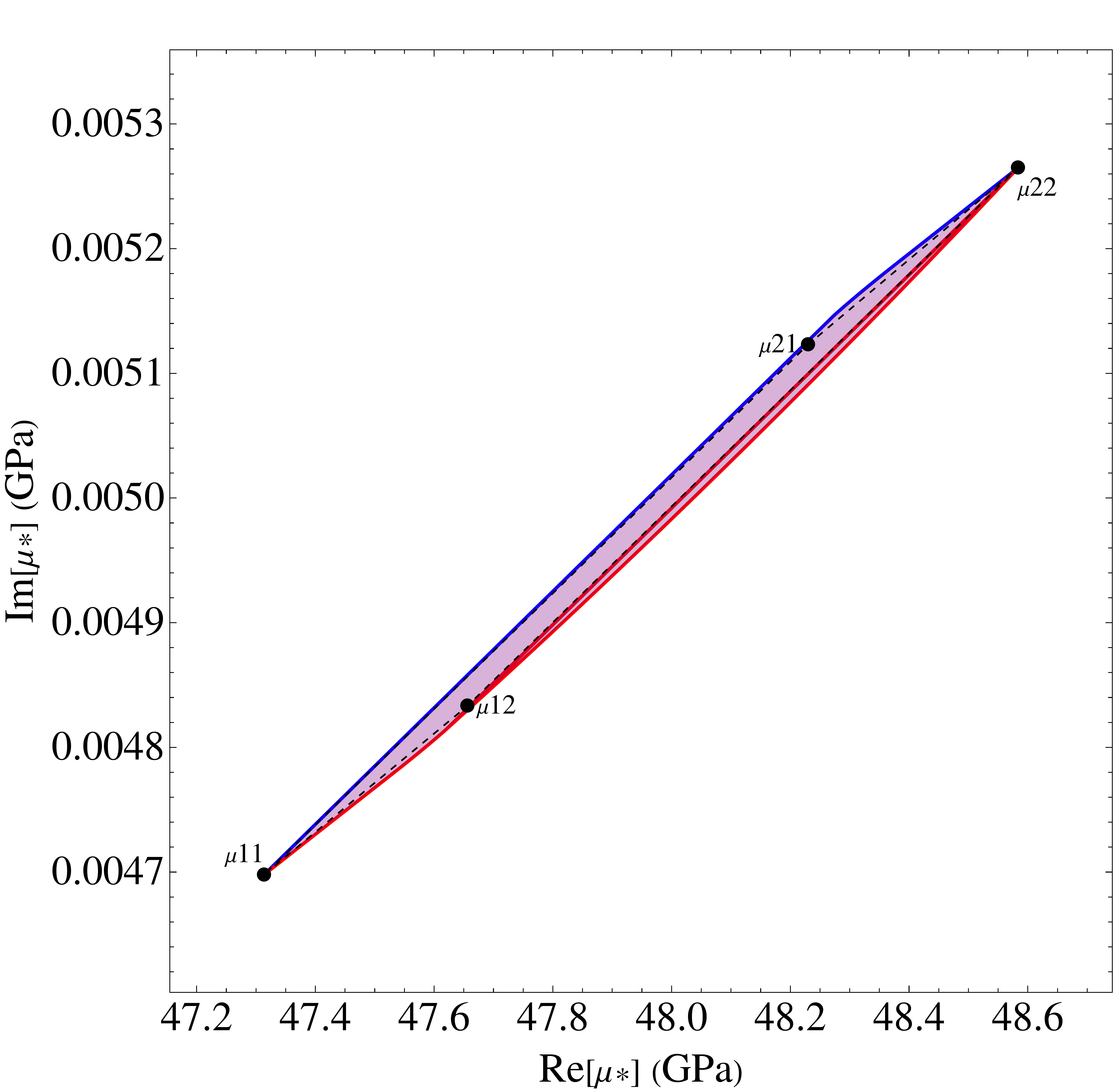}
    \caption{Exact 3D bounding region for the effective complex shear modulus ($\mu^*$) of the $\mathrm{SiO}_2/\mathrm{TiO}_2$ composite at a fixed volume fraction. The shaded purple area represents the strictly permissible domain for any realizable microgeometry. The solid blue and red curves denote the exact mapped inner and outer boundaries, respectively, computed via the continuous parametric algorithm in the complex $Y_\mu$ plane \cite{Milton1997}. The four black dots represent the complex Hashin--Shtrikman--Walpole points ($\mu_{11}$, $\mu_{12}$, $\mu_{21}$, $\mu_{22}$). The dashed black line illustrates the simple quadrilateral connecting these vertices; notably, the exact rigorous bounds exhibit a slight outward convexity, demonstrating that the true permissible region is marginally larger than the polygon defined solely by the extreme points.}
    \label{fig:shear_bounds}
\end{figure}

Figure~\ref{fig:shear_bounds} shows that the exact 3D admissible region is not exactly the quadrilateral defined by the four complex HSW points
$\mu_{11}$, $\mu_{12}$, $\mu_{21}$, and $\mu_{22}$.

In this notation, the double indices represent the specific combinations of the constituent bulk and shear moduli corresponding 
to the theoretical doubly coated sphere microgeometries. 
The exact boundary exhibits a small outward convexity, as expected from the full Milton--Berryman construction
in the auxiliary $Y_\mu$ plane \cite{Milton1997, Milton2002}. However, for the present SiO$_2$/TiO$_2$ system,
the difference between the exact region and the polygon defined by the four extremal points is numerically
very small because the loss angles are of order $10^{-4}$ (Table~\ref{tab:material_properties}).
Therefore, while the four-point construction is not mathematically identical to the exact rigorous
boundary, it provides an excellent low-loss approximation for the present materials system.

The key consequence of Figs.~\ref{fig:refractive_bounds}--\ref{fig:bulk_bounds} is therefore the following.
Strictly speaking, the final full scan is not presented as a fully rigorous coupled optical--mechanical
extremal construction for a single realizable microstructure. However, in the specific low-loss
SiO$_2$/TiO$_2$ case, the exact admissible domains collapse numerically into very narrow regions that are
well approximated by short boundary segments in the optical plane and by the quadrilateral generated by
the four mechanical HSW points. This is the reason why the simplified full phase-space scan remains
physically meaningful and quantitatively informative for the present materials system.

\section{Computational Framework for the Full Phase-Space Scan}
\label{sec:framework}

Having established the methods to compute the opto-mechanical properties for both architectures, this section outlines the numerical protocol used to generate the final dissipation comparison. 
The numerical algorithm, implemented in the Julia programming language \cite{Julia2017},
 evaluates, for each prescribed effective real refractive index $n_{eff}$, the ratio between the dissipated power of the nanolaminate architecture 
($W_{NL}$) and that of an isotropic mixed coating ($W_{ISO}$) under strict optical parity.

The scan is performed over the interval $n_L \le n_{eff} \le n_H$. For each iteration step (i.e., for each target $n_{eff}$), the algorithm executes the following sequence:
\begin{enumerate}
    \item \textbf{Optical Thickness Fixing:} The physical thickness of both coatings is constrained by the quarter-wave condition at the design wavelength $\lambda_0$:
    \begin{equation}
        d_{coat} = \frac{\lambda_0}{4 n_{eff}} .
    \end{equation}
    
    \item \textbf{Nanolaminate Evaluation:} The optical volume fraction $f_H^{NL}$ is extracted using the stratified Wiener relation. This fraction is fed into the Backus averaging formulas in the complex plane to yield $A_{NL}, B_{NL},$ and $D_{NL}$. Because the nanolaminate is perfectly symmetric in-plane, the cross-coupling term vanishes ($R_{mix}^{NL} = 0$). The dissipated power $W_{NL}$ is then computed via Eq.~(\ref{eq:Wdiss}).
    
    \item \textbf{Isotropic Mixture Evaluation:} The volume fraction $f_H^{ISO}$ is computed using three distinct optical closures: L--L, spherical Bruggeman, and generalized Bruggeman ($L=0.40$). 
    For each fraction, the non-linear Self-Consistent Scheme is numerically solved in the complex plane to extract $K_{eff}$ and $\mu_{eff}$. These are mapped to $A_{ISO}, B_{ISO}, D_{ISO},$ 
    and the coupling term $R_{mix}^{ISO}$ (which is strictly non-zero for isotropic mixtures), allowing the computation of $W_{ISO}$ via Eq.~(\ref{eq:Wdiss}).
    
    \item \textbf{Admissible Envelope Construction:} Independent of the specific EMT models, the rigorous Bergman--Milton optical bounds are 
    inverted to obtain the mathematically permissible interval $[f_{H,\min}^{opt}, f_{H,\max}^{opt}]$ for the given real $n_{eff}$. 
The code scans this interval, evaluating the complex mechanical vertices dictated by the HSW bounds. The minimum and maximum dissipation values obtained define the absolute rigorous envelope for the isotropic mixed coating.
\end{enumerate}

For each $n_{eff}$, the code outputs the dissipation ratio $W_{NL}/W_{ISO}$ for the realistic EMT closures, alongside the upper and lower bounds of the rigorous envelope. The results of this full phase-space scan are presented in Fig.~\ref{fig:full_phase_space_scan} and discussed in the next section.

\begin{figure*}[t]
    \centering
    \includegraphics[width=0.82\textwidth]{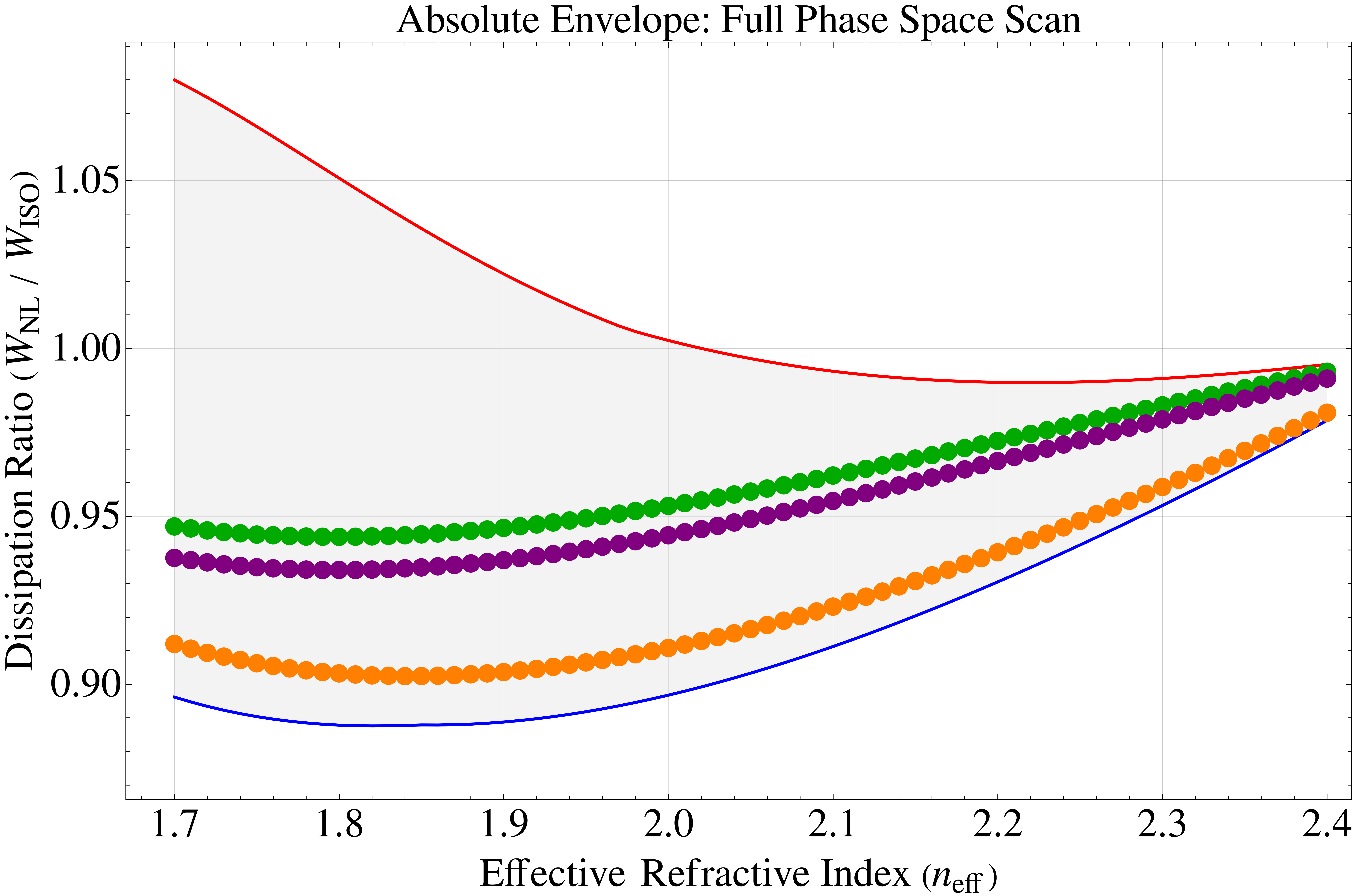}
    \caption{Full phase-space scan of the dissipation ratio $W_{NL}/W_{ISO}$ as a function of the target real effective refractive index $n_{eff}$. 
    The red and blue curves delimit the optically admissible mechanical envelope obtained by scanning the isotropic mixed coating over the range of phase fractions allowed by the Bergman-Milton optical  
    bounds and over the corresponding HS/HSW-like mechanical vertices.
    The orange points correspond to the realistic isotropic mixture obtained by combining the L--L optical closure with the complex self-consistent mechanical solution. 
    The green points show the analogous result obtained with the spherical Bruggeman optical closure, while the purple points correspond to the generalized Bruggeman closure with shape 
    parameter $L=0.40$. Values below unity indicate that the nanolaminate dissipates less power than the corresponding isotropic mixed coating. 
    In the realistic closures considered here, the nanolaminate remains consistently less dissipative over the full interval of $n_{eff}$ explored, while the broader 
    red--blue band shows that this advantage is not absolute over the entire admissible phase space.}
    \label{fig:full_phase_space_scan}
\end{figure*}

\section{Results and Discussion}
\label{sec:results}

The purpose of this section is to translate the effective-medium and bounding results of the previous sections into a direct dissipation comparison between the two coating architectures. We first analyze the realistic isotropic-mixture models introduced in Section~\ref{sec:isotropic_mixtures} and compare them with the nanolaminate baseline through the full phase-space scan of Fig.~\ref{fig:full_phase_space_scan}. We then examine the wider admissible envelope to assess how strongly the conclusion depends on the unknown isotropic microstructure, ultimately discussing the physical mechanisms driving the observed advantage of the ordered layered topology.

\subsection{Thermal-Noise Comparison at Parity of Optical Design}

The final result of the analysis is summarized in Fig.~\ref{fig:full_phase_space_scan}, which reports the dissipation ratio $W_{NL}/W_{ISO}$ as a function of the target real effective 
refractive index $n_{eff}$. The red and blue curves delimit the optically admissible mechanical envelope, while the orange, green, and purple points correspond 
to the three realistic isotropic closures introduced in Section~\ref{subsec:Optical}.

The primary result is that all realistic closures lie strictly below unity over the full interval of $n_{eff}$ explored here. This indicates that, across all realistic isotropic-mixture models considered, the nanolaminate is consistently less dissipative than the corresponding co-sputtered isotropic coating. Among these closures, the L--L model yields the lowest values of $W_{NL}/W_{ISO}$, suggesting that the most finely dispersed isotropic morphology considered here is also the most dissipative relative to the ordered layered architecture.

The spherical Bruggeman and generalized Bruggeman closures are less unfavorable to the isotropic mixture, yet they still preserve a clear dissipative advantage for the nanolaminate. 
In the intermediate refractive-index region ($n_{eff} \sim 1.8$--$2.0$), the reduction in dissipation is on the order of a few percent for Bruggeman-type models, 
approaching the $10\%$ mark for closures more favorable to the nanolaminate architecture, such as the L--L model.

\subsection{Interpretation of the Admissible Envelope}

Figure~\ref{fig:full_phase_space_scan} also illustrates that this advantage is not mathematically absolute over the entire theoretically admissible phase space. 
The upper red branch remains above unity at low and intermediate $n_{eff}$, implying that highly specific, idealized or extremal isotropic microstructures 
could theoretically compete with or exceed the nanolaminate in that corner of the admissible domain.

However, the crucial point is that the rigorous Bergman-Milton and Gibiansky-Milton-Berryman admissible regions in the complex plane collapse into extremely narrow lenses due to the intrinsically low mechanical and optical losses of the constituent materials ($\phi \sim 10^{-4}$, $\kappa \sim 10^{-8}-10^{-3}$). This geometric degeneration drastically restricts the microstructural uncertainty, proving that our phase-space scan remains quantitatively robust and tightly bounds the exact admissible set. Consequently, all physically realistic isotropic closures lie well inside this narrow envelope and remain strictly below unity across the full scan. The superiority of the nanolaminate is therefore not an isolated artifact of a single mixing rule, but a stable, structurally driven feature.

\subsection{The Topological Trade-Off}

The physical origin of this result resolves a highly non-trivial topological trade-off. A priori, it is not obvious that a nanolaminate should outperform an isotropic mixture. From a purely mechanical standpoint, the nanolaminate is actually disadvantaged: its parallel layered structure forces the in-plane shear deformations—excited by the Gaussian laser pressure—to operate at the Voigt upper bound, which maximizes viscoelastic energy dissipation. An isotropic mixture, conversely, distributes shear strain more efficiently. 

However, the nanolaminate possesses a critical optical advantage: it operates at the Wiener upper bound for the in-plane dielectric permittivity. This extreme optical efficiency 
allows the nanolaminate to reach the target real refractive index using a significantly lower volume fraction of the highly dissipative high-index oxide (TiO$_2$)
 compared to an isotropic mixture (as rigorously proven and algebraically certified in Appendices~\ref{app:min_fH} and \ref{app:algebraic_certification}).
Our calculations show that, for the material parameters considered here, this drastic reduction in the volume of lossy material overwhelmingly overcompensates 
 for the mechanical anisotropy penalty, resulting in a net reduction of internal friction.

This theoretical demonstration of a purely topological advantage perfectly complements the known experimental and technological benefits of NLOF. 
As extensively investigated within the gravitational-wave community \cite{Pinto_G1600665, Durante2023, Pinto_G2400526}, 
the geometric confinement inherent in sub-wavelength NLOF frustrates the crystallization of high-index oxides. 
This allows the coatings to withstand significantly higher annealing temperatures without developing scattering defects, which in turn leads to a further,
 substantial reduction in mechanical loss. Furthermore, nanolayering flattens and shifts the low-temperature mechanical loss peaks
 \cite{Pinto_G2001499}, addressing one of the most critical bottlenecks for cryogenic interferometry.

As characterized in Ref.~\cite{Magnozzi2018}, nanoscale confinement
produces a slight reduction and a weak thickness dependence of the
refractive index of the constituent oxides owing to interface and
density effects; the experimentally derived indices adopted here
therefore incorporate these nanoscale optical effects.

The numerical baseline reported in Table~\ref{tab:material_properties} combines established benchmark
parameters for the low-index silica constituent \cite{Flaminio2010,Bodya}, nominal room-temperature mechanical
parameters for amorphous titania, and optical parameters obtained from a refined analysis of the
spectroscopic-ellipsometry data acquired on our nanolayer samples, supported by their structural characterization
\cite{Durante2024,Magnozzi2018}. The quantitative dissipation ratios reported in this work are therefore conditional on this particular
baseline parameter set. The optical component of our argument is independent of the mechanical
loss parameters. Because the nanolaminate operates at the Wiener optical upper bound, it reaches the prescribed effective refractive
index using a strictly smaller volume fraction of the high-index phase, $f_H^{\rm NL}<f_H^{\rm ISO}$, as rigorously proved in Appendix~\ref{app:min_fH}.
For the benchmark dissipative contrast adopted here, $\phi_H>\phi_L$, the reduction of the high-index volume fraction also reduces the amount of the more dissipative constituent.
This provides a physically motivated heuristic explanation for the lower dissipation obtained for the realistic isotropic closures considered in this work, but it does not constitute a general
robustness proof. Within the low-loss approximation, a common rescaling of both constituent loss angles is expected to modify the dissipations of the two architectures by approximately the same factor and therefore to have a much smaller effect on their ratio. By contrast, independent variations of the constituent loss angles, Young's moduli, or Poisson's ratios may modify $W_{\rm NL}/W_{\rm ISO}$. These observations provide only a qualitative sensitivity argument and should not be interpreted as a substitute for a systematic parametric analysis.

\section{Conclusions and Outlook}
\label{sec:conclusions}

In this work, we introduced a fully coupled opto-mechanical framework to evaluate the thermal noise dissipation of 
perfectly ordered periodic nanolaminates compared to disordered isotropic co-sputtered mixtures. By enforcing a strict quarter-wave 
optical thickness constraint and targeting the exact same effective real refractive index, we ensured an unbiased and physically consistent 
comparison between the two architectures for the baseline $\mathrm{SiO_2/TiO_2}$ system.

By evaluating the exact Wiener--Backus limits for the layered medium against realistic Effective Medium Theories (Lorentz--Lorenz and Bruggeman closures coupled with a complex Self-Consistent Scheme) and mapping their absolute variational bounds, we demonstrated that the nanolaminate systematically exhibits lower thermal noise dissipation ($W_{NL}/W_{ISO} < 1$). 

The physical mechanism driving this performance is a favorable topological trade-off. 
The nanolaminate operates at the Wiener upper bound for the in-plane dielectric permittivity, maximizing optical efficiency. 
As mathematically proven and algebraically certified in Appendices~\ref{app:min_fH} and \ref{app:algebraic_certification}, this efficiency 
allows the ordered structure to reach the target refractive index using a strictly lower volume fraction of the highly dissipative high-index phase. 
We showed that this reduction in the volume of lossy material more than compensates for the inherent mechanical penalty associated with in-plane shear deformations.

Nonetheless, we must acknowledge certain physical limitations of the present idealized theoretical framework. First, our thermodynamic analysis assumes that the constituent amorphous phases ($\mathrm{SiO_2}$ and $\mathrm{TiO_2}$) preserve their independent structural, optical, and dissipative identities when co-deposited or layered. In real co-sputtered mixtures, nanoscale structural rearrangements, localized coordination changes, or chemical interactions between the different oxides might alter their intrinsic loss angles, potentially causing the composite to depart from the predicted thermodynamic bounds. Second, our nanolaminate model relies on perfectly sharp, step-like interfaces between the alternating sub-wavelength layers. 
In practice, a sub-nanometer interdiffusion (intermixing) zone at the boundaries inevitably occurs during deposition and subsequent thermal annealing. 
Although experimental characterizations show that such interdiffused regions are extremely narrow \cite{Pierro_ScienceDirect2023}, their presence could introduce 
local gradient-index optical effects and shear stress relaxation states that are neglected in our exact Wiener--Backus treatment. 
Nevertheless, these idealizations do not undermine the physical validity of our approach; rather, they establish a rigorous theoretical benchmark---an ultimate performance limit to strive for in the practical development of composite optical coatings.

A further limitation of the present analysis is the use of a single loss angle for each constituent phase, implemented through purely real Poisson's ratios. Introducing independent bulk and
shear loss angles, $\phi_{K,i}\neq\phi_{\mu,i}$, would generally result in complex constituent Poisson's ratios and would require the viscoelastic bounds and the corresponding Brownian dissipations to be
recomputed. In the zero-loss limit, the admissible complex-modulus regions reduce to the classical Hashin--Shtrikman--Walpole intervals on the real axis. Independent imaginary components generally open
these intervals into two-dimensional regions bounded by circular arcs or, for the general shear-modulus problem, by intersections of circular bounding regions \cite{Gibiansky1993,Milton1997}. Since the loss angles
relevant to the present optical coatings are small, these admissible regions are expected to remain narrow in absolute terms, although their width, orientation, and detailed geometry may change.

We conjecture that, for the benchmark $\mathrm{SiO_2/TiO_2}$ system and the realistic isotropic closures considered here, the inequality $W_{\rm NL}/W_{\rm ISO}<1$ remains valid under sufficiently small departures from the single-loss-angle approximation, provided that all constituent bulk and shear loss angles remain in the low-loss regime. A quantitative assessment of this conjecture and, more generally, of the robustness of the present conclusions requires a systematic parametric investigation of the constituent optical and mechanical properties, supported by reliable constituent-specific experimental data, particularly independent measurements of the bulk and shear loss angles. This combined theoretical, computational, and experimental extension represents an important direction for future work.

Under this perspective, these theoretical findings indicate that synthesizing optical materials as highly ordered 1D superlattices provides a fundamental,
structurally driven reduction in internal friction. Combined with their experimentally proven ability to frustrate crystallization and mitigate low-temperature 
mechanical loss peaks, nanolaminates emerge as a highly promising architectural pathway, setting a clear design target for the development of ultra-low noise 
mirrors in next-generation (3G) gravitational-wave observatories, such as the Einstein Telescope \cite{ETsite} and Cosmic Explorer \cite{CEsite}.

\section*{Acknowledgements}
The authors are deeply grateful to I. M. Pinto, who provided the original inspiration for this research,  
for his continuous encouragement and invaluable technical insights throughout this work. 
We also extend our sincere appreciation to G. Cagnoli and E. Cesarini for their constructive feedback. 
Furthermore, we acknowledge the Virgo CRD and LIGO Optics Working Groups for the stimulating scientific exchanges that contributed to this research. 
This work was partially supported by INFN Sezione di Napoli, Gruppo Collegato di Salerno.

\appendix
\section{Minimum high-index fraction}
\label{app:min_fH}

This appendix proves a statement stronger than the usual intermediate-index property of passive mixtures. 
Specifically, in the passive low-loss complex regime relevant to the optical oxides considered here, the Wiener-parallel nanolaminate 
reaches any strictly interior prescribed real refractive index with a strictly smaller high-index volume fraction than 
any admissible macroscopically isotropic morphology made of the same two constituents.

We denote the complex permittivities of the high- and low-index phases by
\begin{equation}
    \varepsilon_H=\varepsilon_H' + i\varepsilon_H'',
    \qquad
    \varepsilon_L=\varepsilon_L' + i\varepsilon_L'',
    \label{eq:app_perm_def}
\end{equation}
with
\begin{equation}
    \varepsilon_H' > \varepsilon_L' >0,
    \qquad
    \varepsilon_H'',\varepsilon_L''\ge 0.
    \label{eq:app_passive_constituents}
\end{equation}
We also set
\begin{equation}
    f_L=1-f_H,
    \,
    \Delta:=\varepsilon_H-\varepsilon_L
    =\Delta' + i\Delta'',
    \,
    \Delta'=\varepsilon_H'-\varepsilon_L'>0.
    \label{eq:app_delta_def}
\end{equation}
No sign assumption is imposed on \(\Delta''\). In particular, the proof below
remains valid even if the high-index phase is optically less lossy than the
low-index phase.

We measure the optical-loss strength by
\begin{equation}
    \eta_{\max}:=
    \max\left\{
    \frac{\varepsilon_L''}{\varepsilon_L'},
    \frac{\varepsilon_H''}{\varepsilon_H'},
    \frac{|\varepsilon_H''-\varepsilon_L''|}
         {\varepsilon_H'-\varepsilon_L'}
    \right\}.
    \label{eq:app_eta_max_def}
\end{equation}
The sign-independent low-loss hypothesis used in this appendix is
\begin{equation}
    \eta_{\max}< \frac14.
    \label{eq:app_loss_hyp}
\end{equation}
The factor \(1/4\) is a sign-independent safety margin controlling possible
upward excursions of the Bergman--Milton lens when \(\Delta''\) is allowed to
have either sign. For the SiO$_2$/TiO$_2$ data used in this paper, this
hypothesis is satisfied by a very large margin.

The nanolaminate in-plane permittivity is the Wiener upper value
\begin{equation}
    \varepsilon_k(f_H)=f_L\varepsilon_L+f_H\varepsilon_H
    =\varepsilon_L+f_H\Delta.
    \label{eq:app_wiener_value}
\end{equation}
For the same volume fraction, the effective permittivity of any
macroscopically isotropic two-phase composite lies in the Bergman--Milton
admissible lens \(\mathcal L(f_H)\). Its boundary consists of the two arcs
\begin{equation}
\label{eq:app_BM_boundary_arcs}
\begin{aligned}
    \varepsilon_{b}^{LH}(u)
    &=
    \varepsilon_k(f_H)
    -
    \frac{f_Lf_H\Delta^2}
         {3\left[u\varepsilon_L+(1-u)\varepsilon_H\right]}, \\
    &\quad \text{for } \tfrac{f_H}{3}\le u\le 1-\tfrac{f_L}{3}, \\[6pt]
    \varepsilon_{b}^{HL}(v)
    &=
    \varepsilon_k(f_H)
    -
    \frac{f_Lf_H\Delta^2}
         {3\left[v\varepsilon_H+(1-v)\varepsilon_L\right]}, \\
    &\quad \text{for } \tfrac{f_L}{3}\le v\le 1-\tfrac{f_H}{3}.
\end{aligned}
\end{equation}
These are the same boundary arcs used in the main text, reported here for ease of reference.

\begin{lemma}[Polar identity for the real refractive index]
\label{lem:app_polar_identity}
Let \(z=x+iy\), with \(x>0\) and \(y\ge0\), and let the square root be taken on
the principal branch. Then
\begin{equation}
    2\left(\operatorname{Re}\sqrt z\right)^2=|z|+x.
    \label{eq:app_polar_identity}
\end{equation}
\end{lemma}

\begin{proof}
Write \(z=|z|e^{i\theta}\), with \(\theta\in[0,\pi/2)\). Then
\[
    \sqrt z=\sqrt{|z|}\,e^{i\theta/2},
    \qquad
    \operatorname{Re}\sqrt z=\sqrt{|z|}\cos\frac{\theta}{2}.
\]
Therefore
\[
    2\left(\operatorname{Re}\sqrt z\right)^2
    =2|z|\cos^2\frac{\theta}{2}
    =|z|(1+\cos\theta)=|z|+x.
\]
\end{proof}

\begin{lemma}[Strict monotonicity of the nanolaminate index]
\label{lem:app_NL_monotonicity}
Under the low-loss hypothesis \eqref{eq:app_loss_hyp}, the function
\begin{equation}
    F_{\rm NL}(f_H):=
    \operatorname{Re}\sqrt{\varepsilon_k(f_H)}
    \label{eq:app_FNL_def}
\end{equation}
is strictly increasing for \(0<f_H<1\).
\end{lemma}

\begin{proof}
Let
\[
    \varepsilon_k(f_H)=x_k+iy_k=|\varepsilon_k|e^{i\theta_k},
    \qquad
    \theta_k\in[0,\pi/2).
\]
By Lemma~\ref{lem:app_polar_identity},
\[
    2F_{\rm NL}^2=|\varepsilon_k|+x_k.
\]
Since \(d\varepsilon_k/df_H=\Delta\), differentiation gives
\begin{equation}
\begin{aligned}
    4F_{\rm NL}\frac{dF_{\rm NL}}{df_H}
    &=
    \frac{d|\varepsilon_k|}{df_H} + \frac{dx_k}{df_H} \\
    &= \Delta'\cos\theta_k+ \Delta''\sin\theta_k+ \Delta' \\
    &= \Delta'\left(1+ \cos\theta_k+\eta\sin\theta_k\right),
\end{aligned}
\label{eq:app_FNL_derivative}
\end{equation}
where
\begin{equation}
    \eta:=\frac{\Delta''}{\Delta'}.
    \label{eq:app_eta_def}
\end{equation}
By \eqref{eq:app_eta_max_def}, \(|\eta|\le \eta_{\max}<1\). Moreover,
\[
    1+\cos\theta_k>\sin\theta_k \qquad \text{for }\theta_k\in[0,\pi/2).
\]
Hence
\[
    1+\cos\theta_k+\eta\sin\theta_k \ge 1+\cos\theta_k-|\eta|\sin\theta_k >0.
\]
Since \(\Delta'>0\) and \(F_{\rm NL}>0\), Eq.~\eqref{eq:app_FNL_derivative}
gives 
\[
dF_{\rm NL}/df_H>0
\].
\end{proof}

\begin{lemma}[Dominant horizontal separation of the boundary]
\label{lem:app_boundary_gap}
Let \(0<f_H<1\). Denote by \(\varepsilon_k=f_L\varepsilon_L+f_H\varepsilon_H\)
the Wiener upper permittivity, and let
\(\varepsilon_b\in\partial\mathcal L(f_H)\) be any point on the
Bergman--Milton boundary. If
\[
    \varepsilon_k=x_k+iy_k,
    \qquad
    \varepsilon_b=x_b+iy_b,
\]
and
\[
    \delta x:=x_k-x_b,
    \qquad
    \delta y:=y_b-y_k,
\]
then, under the low-loss hypothesis \eqref{eq:app_loss_hyp},
\begin{equation}
    \delta x>0,
    \qquad
    |\delta y|
    \le
    \frac{3\eta_{\max}}{1-3\eta_{\max}^2}\,\delta x
    <
    4\eta_{\max} \delta x.
    \label{eq:app_boundary_gap}
\end{equation}
Equivalently, the boundary of the isotropic admissible lens remains separated
from the Wiener point by a strictly positive real gap, while its vertical
excursion is only a low-loss correction of order \(\eta_{\max}\).
\end{lemma}

\begin{proof}
Both boundary arcs in \eqref{eq:app_BM_boundary_arcs} can be written in the
unified form
\begin{equation}
    \varepsilon_b = \varepsilon_k - \frac{f_Lf_H\Delta^2}{D},
    \label{eq:app_boundary_unified_form}
\end{equation}
where
\begin{equation}
    D = 3\left[u\varepsilon_L+(1-u)\varepsilon_H\right] \quad\text{or}\quad D = 3\left[v\varepsilon_H+(1-v)\varepsilon_L\right].
    \label{eq:app_boundary_denominators}
\end{equation}
The parameters \(u\) and \(v\) range over the intervals stated in \eqref{eq:app_BM_boundary_arcs}. In both cases, the coefficients of $D/3$ are
non-negative and sum to one. Thus $D/3$ is a convex combination of the two constituent permittivities. Consequently, writing \(D=D'+iD''\),
\begin{equation}
    D'>0, \qquad 0\le D''\le \eta_{\max}D'.
    \label{eq:app_D_bounds}
\end{equation}
Define
\begin{equation}
    \eta:=\frac{\Delta''}{\Delta'}, \qquad \gamma:=\frac{D''}{D'}, \qquad C:= \frac{f_Lf_HD'(\Delta')^2}{|D|^2}>0.
    \label{eq:app_eta_gamma_C_def}
\end{equation}
Then, by \eqref{eq:app_eta_max_def} and \eqref{eq:app_D_bounds},
\begin{equation}
    |\eta|\le \eta_{\max}, \qquad 0\le\gamma\le\eta_{\max}.
    \label{eq:app_eta_gamma_bounds}
\end{equation}
Using \eqref{eq:app_boundary_unified_form} and separating real and imaginary parts gives the exact identities
\begin{equation}
    \delta x = C\left(1-\eta^2+2\gamma\eta\right),
    \label{eq:app_deltax_identity}
\end{equation}
and
\begin{equation}
    \delta y = C\left[\gamma(1-\eta^2)-2\eta\right].
    \label{eq:app_deltay_identity}
\end{equation}
From \eqref{eq:app_eta_gamma_bounds},
\begin{equation}
    \delta x \ge C\left(1-3\eta_{\max}^2\right)>0,
    \label{eq:app_deltax_lower_bound}
\end{equation}
and, since \(|\eta|<1\),
\begin{equation}
\begin{aligned}
    |\delta y| &\le C\left(\gamma(1-\eta^2)+2|\eta|\right) \\
    &\le C\left(\eta_{\max}+2\eta_{\max}\right) = 3\eta_{\max}C.
\end{aligned}
\label{eq:app_deltay_upper_bound}
\end{equation}
Combining \eqref{eq:app_deltax_lower_bound} and \eqref{eq:app_deltay_upper_bound} yields
\[
    |\delta y| \le \frac{3\eta_{\max}}{1-3\eta_{\max}^2}\,\delta x.
\]
Since \(\eta_{\max}<1/4\), the prefactor satisfies
\[
    \frac{3}{1-3\eta_{\max}^2}<4,
\]
and the final inequality in \eqref{eq:app_boundary_gap} follows.
\end{proof}

\begin{lemma}[From boundary separation to refractive-index separation]
\label{lem:app_boundary_index_gap}
For every \(0<f_H<1\) and every boundary point \(\varepsilon_b\in\partial\mathcal L(f_H)\),
\begin{equation}
    \operatorname{Re}\sqrt{\varepsilon_b} < \operatorname{Re}\sqrt{\varepsilon_k(f_H)}.
    \label{eq:app_boundary_index_gap}
\end{equation}
\end{lemma}

\begin{proof}
Write $\varepsilon_k=x_k+iy_k$, $\varepsilon_b=x_b+iy_b$, and define $\delta x=x_k-x_b$, $\delta y=y_b-y_k$.
We first prove the identity
\begin{equation}
\begin{aligned}
&2\left[ \left(\operatorname{Re}\sqrt{\varepsilon_k}\right)^2 - \left(\operatorname{Re}\sqrt{\varepsilon_b}\right)^2 \right] \left(|\varepsilon_k|+|\varepsilon_b|\right) \\
&\qquad = \delta x\left(|\varepsilon_k|+|\varepsilon_b|+x_k+x_b\right) - \left(y_b^2-y_k^2\right).
\end{aligned}
\label{eq:app_square_root_difference_identity}
\end{equation}
By Lemma~\ref{lem:app_polar_identity},
\[
    2\left[ \left(\operatorname{Re}\sqrt{\varepsilon_k}\right)^2 - \left(\operatorname{Re}\sqrt{\varepsilon_b}\right)^2 \right] = \left(|\varepsilon_k|+x_k\right) - \left(|\varepsilon_b|+x_b\right).
\]
Since \(\delta x=x_k-x_b\), this becomes $|\varepsilon_k|-|\varepsilon_b|+ \delta x$. Multiplying by \(|\varepsilon_k|+|\varepsilon_b|\), we obtain
\begin{align*}
&2\left[ \left(\operatorname{Re}\sqrt{\varepsilon_k}\right)^2 - \left(\operatorname{Re}\sqrt{\varepsilon_b}\right)^2 \right] \left(|\varepsilon_k|+|\varepsilon_b|\right) \\
&\qquad = \left(|\varepsilon_k|-|\varepsilon_b|\right)\left(|\varepsilon_k|+|\varepsilon_b|\right) + \delta x\left(|\varepsilon_k|+|\varepsilon_b|\right).
\end{align*}
Using
\[
    \left(|\varepsilon_k|-|\varepsilon_b|\right)\left(|\varepsilon_k|+|\varepsilon_b|\right) = |\varepsilon_k|^2-|\varepsilon_b|^2,
\]
and
\[
    |\varepsilon_k|^2=x_k^2+y_k^2, \qquad |\varepsilon_b|^2=x_b^2+y_b^2,
\]
we get
\[
    |\varepsilon_k|^2-|\varepsilon_b|^2 = \delta x(x_k+x_b)- \left(y_b^2-y_k^2\right).
\]
Therefore,
\begin{align*}
&2\left[ \left(\operatorname{Re}\sqrt{\varepsilon_k}\right)^2 - \left(\operatorname{Re}\sqrt{\varepsilon_b}\right)^2 \right] \left(|\varepsilon_k|+|\varepsilon_b|\right) \\
&\qquad = \delta x(x_k+x_b)- \left(y_b^2-y_k^2\right) + \delta x\left(|\varepsilon_k|+|\varepsilon_b|\right) \\
&\qquad = \delta x\left(|\varepsilon_k|+|\varepsilon_b|+x_k+x_b\right) - \left(y_b^2-y_k^2\right),
\end{align*}
which proves \eqref{eq:app_square_root_difference_identity}.

By Lemma~\ref{lem:app_boundary_gap}, \(\delta x>0\). If \(y_b\le y_k\), then, since boundary points are passive, \(y_b\ge0\), and of course \(y_k\ge0\). Thus $y_b^2-y_k^2\le0$, and the right-hand side of \eqref{eq:app_square_root_difference_identity} is strictly positive.

It remains only to consider the upward-bulge case \(y_b>y_k\). Then \(\delta y>0\), and Lemma~\ref{lem:app_boundary_gap} gives
\begin{equation}
    0<\delta y<4\eta_{\max}\delta x.
    \label{eq:app_upward_delta_y_bound}
\end{equation}
Moreover, by the definition of \(\eta_{\max}\),
\begin{equation}
    y_k = f_L\varepsilon_L''+f_H\varepsilon_H'' \le \eta_{\max} \left(f_L\varepsilon_L'+f_H\varepsilon_H'\right) = \eta_{\max}x_k.
    \label{eq:app_yk_low_loss_bound}
\end{equation}
The Bergman--Milton lens lies in the right half-plane under the present passive low-loss assumptions, so \(x_b>0\). Hence
\begin{equation}
    0<\delta x<x_k.
    \label{eq:app_deltax_less_than_xk}
\end{equation}
Combining \eqref{eq:app_upward_delta_y_bound}, \eqref{eq:app_yk_low_loss_bound}, and \eqref{eq:app_deltax_less_than_xk},
\begin{equation}
    y_b+y_k = 2y_k+\delta y < 2\eta_{\max}x_k+4\eta_{\max}\delta x < 6\eta_{\max}x_k.
    \label{eq:app_yb_plus_yk_bound}
\end{equation}
Therefore
\begin{equation}
    y_b^2-y_k^2 = \delta y\,(y_b+y_k) < 24\eta_{\max}^2x_k\delta x.
    \label{eq:app_y_square_excess_bound}
\end{equation}
On the other hand,
\begin{equation}
    |\varepsilon_k|+|\varepsilon_b|+x_k+x_b = \left(|\varepsilon_k|+x_k\right) + \left(|\varepsilon_b|+x_b\right) \ge 2x_k.
    \label{eq:app_horizontal_term_lower_bound}
\end{equation}
Substituting \eqref{eq:app_y_square_excess_bound} and \eqref{eq:app_horizontal_term_lower_bound} into \eqref{eq:app_square_root_difference_identity} gives
\begin{equation}
\begin{aligned}
&\delta x\left(|\varepsilon_k|+|\varepsilon_b|+x_k+x_b\right) - \left(y_b^2-y_k^2\right) \\
&\qquad > 2x_k\delta x - 24\eta_{\max}^2x_k\delta x = 2x_k\delta x\left(1-12\eta_{\max}^2\right).
\end{aligned}
    \label{eq:app_final_positive_bound}
\end{equation}
Since \(\eta_{\max}<1/4\), the factor \(1-12\eta_{\max}^2\) is positive. Hence the right-hand side of \eqref{eq:app_square_root_difference_identity} is
strictly positive also in the upward-bulge case. Thus $(\operatorname{Re}\sqrt{\varepsilon_k})^2 > (\operatorname{Re}\sqrt{\varepsilon_b})^2$, and \eqref{eq:app_boundary_index_gap} follows.
\end{proof}

\begin{lemma}[Dominance over the whole isotropic admissible lens]
\label{lem:app_lens_dominance}
For every \(0<f_H<1\) and every admissible isotropic effective permittivity \(\varepsilon_{\rm eff}^{\rm ISO}\in\mathcal L(f_H)\),
\begin{equation}
    \operatorname{Re}\sqrt{\varepsilon_{\rm eff}^{\rm ISO}} < \operatorname{Re}\sqrt{\varepsilon_k(f_H)}.
    \label{eq:app_lens_dominance}
\end{equation}
\end{lemma}

\begin{proof}
The principal square root is analytic in the right half-plane. Hence $\Phi(z):=\operatorname{Re}\sqrt z$
is harmonic there. The Bergman--Milton lens \(\mathcal L(f_H)\) is compact and lies in the right half-plane under the passive low-loss assumptions above.
Therefore, by the maximum principle for harmonic functions, the maximum of \(\Phi\) over \(\mathcal L(f_H)\) is attained on \(\partial\mathcal L(f_H)\). Lemma~\ref{lem:app_boundary_index_gap} shows that
every boundary value is strictly smaller than \(\operatorname{Re}\sqrt{\varepsilon_k(f_H)}\). The same strict inequality therefore holds throughout the whole admissible lens.
\end{proof}

\begin{theorem}[Strict inequality of the required high-index volume fractions]
\label{thm:app_fraction_inequality}
Let \(n_*\) be a strictly interior target real refractive index, $n_L < n_* < n_H$, attained by both architectures
with volume fractions in \((0,1)\). Let \(f_H^{\rm NL}\) be the high-index fraction required by the nanolaminate and let \(f_H^{\rm ISO}\) be the
high-index fraction required by an admissible macroscopically isotropic morphology:
\begin{equation}
    \operatorname{Re}\sqrt{\varepsilon_k(f_H^{\rm NL})}=n_*,
    \qquad
    \operatorname{Re}\sqrt{\varepsilon_{\rm eff}^{\rm ISO}(f_H^{\rm ISO})} = n_*.
    \label{eq:app_target_index_equalities}
\end{equation}
Then, under the low-loss hypothesis \eqref{eq:app_loss_hyp},
\begin{equation}
    f_H^{\rm NL}<f_H^{\rm ISO}.
    \label{eq:app_fraction_inequality}
\end{equation}
\end{theorem}

\begin{proof}
By Lemma~\ref{lem:app_lens_dominance}, evaluated at \(f_H=f_H^{\rm ISO}\),
\[
    n_* = \operatorname{Re}\sqrt{\varepsilon_{\rm eff}^{\rm ISO}(f_H^{\rm ISO})} < \operatorname{Re}\sqrt{\varepsilon_k(f_H^{\rm ISO})} = F_{\rm NL}(f_H^{\rm ISO}).
\]
By definition of \(f_H^{\rm NL}\), $n_*=F_{\rm NL}(f_H^{\rm NL})$. Therefore $F_{\rm NL}(f_H^{\rm NL})<F_{\rm NL}(f_H^{\rm ISO})$.
Since \(F_{\rm NL}\) is strictly increasing by Lemma~\ref{lem:app_NL_monotonicity}, the inequality of the arguments follows: $f_H^{\rm NL}<f_H^{\rm ISO}$.
\end{proof}

\noindent {\bf Final Remark:} Optical origin of the high-index saving\\
The ordered nanolaminate uses the Wiener upper optical response to minimize
the amount of high-index phase required to reach a prescribed real refractive
index. Since TiO$_2$ is the mechanically lossier constituent in the material
system considered here, the strict inequality
\[
    f_H^{\rm NL}<f_H^{\rm ISO}
\]
provides the optical-topological mechanism behind the reduction in high-index
material content used in the main-text dissipation comparison.

\section{Algebraic certification of the boundary separation}
\label{app:algebraic_certification}

We begin with a general boundary criterion that extends the argument of
Appendix~A beyond the explicit low-loss estimate. It shows that the desired
separation follows once a suitable boundary maximum is certified to be
strictly negative.

\begin{proposition}[Certified boundary criterion beyond the explicit low-loss bound]
Fix the constituent permittivities \(\varepsilon_L,\varepsilon_H\) and a
volume fraction \(0<f_H<1\). Define
\[
    M(f_H) := \max_{\varepsilon\in\partial\mathcal L(f_H)} \operatorname{Re}\sqrt{\varepsilon} - \operatorname{Re}\sqrt{\varepsilon_k(f_H)} .
\]
If $M(f_H)<0$ for all \(f_H\) in the interval \(I = (0,1)\), then
\[
    \operatorname{Re}\sqrt{\varepsilon_{\mathrm{eff}}^{\mathrm{ISO}}} < \operatorname{Re}\sqrt{\varepsilon_k(f_H)}
\]
for every admissible isotropic effective permittivity \(\varepsilon_{\mathrm{eff}}^{\mathrm{ISO}}\in\mathcal L(f_H)\) and every \(f_H\in I\).
\end{proposition}

\begin{proof}
The function $\Phi(z)=\operatorname{Re}\sqrt z$
is harmonic in the right half-plane ($\operatorname{Re}[z]>0$) . 
Since the Bergman--Milton lens \(\mathcal L(f_H)\) is compact and lies in this half-plane for passive
dielectrics in the regime considered here, the maximum principle gives
\[
    \max_{\varepsilon\in\mathcal L(f_H)} \operatorname{Re}\sqrt{\varepsilon} = \max_{\varepsilon\in\partial\mathcal L(f_H)} \operatorname{Re}\sqrt{\varepsilon}.
\]
Thus, if \(M(f_H)<0\), every point of the admissible lens satisfies $\operatorname{Re}\sqrt{\varepsilon} < \operatorname{Re}\sqrt{\varepsilon_k(f_H)}$.
\end{proof}

\begin{remark}[Computer-assisted certification]
The criterion above is independent of the explicit sufficient condition \(\eta_{\max}<1/4\). Since the Bergman--Milton boundary is the union of two
fractional-linear arcs, the boundary maximum \(M(f_H)\) can be reduced to one-dimensional analytic problems on compact intervals. This makes the
criterion suitable for certified computer-assisted verification. 
Such a verification may be carried out either via rigorous numerical enclosures based on interval analysis \cite{Moore2009, Julia2017}, 
or by exact real-algebraic quantifier elimination after rewriting the boundary-separation condition as the non-intersection of algebraic curves.
The latter approach is used below for the SiO$_2$/TiO$_2$ optical constants considered in this work.
\end{remark}

\subsection{Computer aided Verification}
We now apply the real-algebraic route to the specific SiO$_2$/TiO$_2$
optical constants used throughout the paper. The goal is to certify directly,
without using the explicit low-loss estimate of Appendix~\ref{app:min_fH},
that the Bergman--Milton boundary does not intersect the level curve
\[
    \operatorname{Re}\sqrt z=\operatorname{Re}\sqrt{\varepsilon_k(f_H)}
\]
for any \(0<f_H<1\). 
This provides a material-specific certification of the boundary separation
\[
    \operatorname{Re}\sqrt{\varepsilon_b} < \operatorname{Re}\sqrt{\varepsilon_k(f_H)}, \qquad \varepsilon_b\in\partial\mathcal L(f_H).
\]
Let $z=x+iy$ and $\varepsilon_k(f_H)=x_k+iy_k$. By the polar identity,
\[
    2\left(\operatorname{Re}\sqrt z\right)^2=|z|+x.
\]
Therefore the level curve $\operatorname{Re}\sqrt z=\operatorname{Re}\sqrt{\varepsilon_k(f_H)}$ is equivalently given by
\[
    |z|+x=|\varepsilon_k|+x_k.
\]
Introducing $c:=|\varepsilon_k|+x_k>0$, one obtains the algebraic relation
\[
    c^2-2x_kc-y_k^2=0,
\]
and the level curve becomes the parabola
\[
    y^2-c^2+2cx=0.
\]

Each Bergman--Milton boundary branch has the rational form
\[
    \varepsilon_b(t) = \varepsilon_k(f_H) - \frac{f_Lf_H\Delta^2}{D_B(t)},
\]
where \(t=u\) on the \(LH\) branch and \(t=v\) on the \(HL\) branch.
For each branch \(B\in\{LH,HL\}\) the denominator \(D_B(t)\) is affine in \(t\), with
\[
    D_{LH}(t)=3\bigl[t\,\varepsilon_L+(1-t)\varepsilon_H\bigr],\quad D_{HL}(t)=3\bigl[t\,\varepsilon_H+(1-t)\varepsilon_L\bigr].
\]
Writing $D_B(t)=D_r(t)+iD_i(t)$, $Q_B(t):=D_r(t)^2+D_i(t)^2$, and $f_Lf_H\Delta^2=A_r+iA_i$,
one obtains $\varepsilon_b(t)=x_b(t)+iy_b(t)$, with
\[
    x_b(t)=\frac{N_x(t)}{Q_B(t)}, \qquad y_b(t)=\frac{N_y(t)}{Q_B(t)},
\]
where
\[
    N_x(t) = x_k Q_B(t)-\left[A_rD_r(t)+A_iD_i(t)\right],
\]
and
\[
    N_y(t) = y_k Q_B(t)-\left[A_iD_r(t)-A_rD_i(t)\right].
\]

After clearing denominators, the condition that a boundary point reaches
or crosses the level curve is $P_B(t,c,f_H)\ge 0$, where
\[
    P_B(t,c,f_H) = N_y(t)^2-c^2Q_B(t)^2+2cN_x(t)Q_B(t).
\]
Indeed,
\[
    \frac{P_B(t,c,f_H)}{Q_B(t)^2} = y_b(t)^2-c^2+2cx_b(t),
\]
so \(P_B(t,c,f_H)\ge0\) is equivalent to $\operatorname{Re}\sqrt{\varepsilon_b(t)} \ge \operatorname{Re}\sqrt{\varepsilon_k(f_H)}$.

The algebraic certification consists of deciding, for each boundary branch,
the real existential formula:
\begin{align*}
\exists\, f_H,t,c: \quad & 0<f_H<1,\quad t\in I_B(f_H),\quad c>0, \\
& C(c,f_H) = c^2-2x_kc-y_k^2=0, \quad P_B(t,c,f_H)\ge0.
\end{align*}
Here $I_B$ can be one of the following closed intervals:
\[
    I_{LH}(f_H) = \left[ \frac{f_H}{3}, 1-\frac{f_L}{3} \right], \qquad I_{HL}(f_H) = \left[ \frac{f_L}{3}, 1-\frac{f_H}{3} \right].
\]
Using exact rational input values for complex refractive indices
\[
    n_L = \frac{145}{100}, \quad \kappa_L = \frac{1}{10^8}, \quad n_H = \frac{247}{100}, \quad \kappa_H = \frac{1}{1000},
\]
with $\varepsilon_i'=n_i^2-\kappa_i^2$ and $\varepsilon_i''=2n_i\kappa_i$,
this real-algebraic decision problem was evaluated using \textit{Wolfram Mathematica} by real-algebraic quantifier elimination.
For both the \(LH\) and \(HL\) branches, the existential formula returned \texttt{False}. Hence no point of either
Bergman--Milton boundary branch satisfies $\operatorname{Re}\sqrt{\varepsilon_b} \ge \operatorname{Re}\sqrt{\varepsilon_k(f_H)}$.
Therefore, $\operatorname{Re}\sqrt{\varepsilon_b} < \operatorname{Re}\sqrt{\varepsilon_k(f_H)}$ for every \(0<f_H<1\) and every \(\varepsilon_b\in\partial\mathcal L(f_H)\).

As a complementary geometric certification, the auxiliary variable \(c\) was formally eliminated by computing 
the Sylvester resultant \cite{BPR2006, CLO2015} of the two polynomials:
\[
    R_B(t,f_H) = \operatorname{Res}_c \Bigl( P_B(t,c,f_H), \,\, C(c,f_H)=c^2-2x_kc-y_k^2 \Bigr).
\]
The corresponding real-root problem,
\[
    R_B(t,f_H)=0, \qquad 0<f_H<1, \qquad t\in I_B(f_H),
\]
was also decided by quantifier elimination and returned \texttt{False} for both branches. 
This class of problems is deterministically solvable via the Cylindrical Algebraic Decomposition (CAD) algorithm, originally introduced 
by Collins \cite{Collins1975} and extensively developed for modern real algebraic geometry \cite{BPR2006}.
This confirms algebraically that the two Bergman--Milton arcs do not intersect the level parabola.

\subsection*{Pseudocode of the algebraic certification}

The algorithm below certifies algebraically that no point of either
Bergman--Milton boundary arc can attain a real refractive index equal
to or greater than that of the Wiener upper bound, for any volume
fraction $f_H\in(0,1)$.

Two auxiliary objects drive the construction.
The \emph{level-curve constraint} 
\[
C(c,f_H):=c^2-2x_kc-y_k^2=0 \quad \text{with}\quad c>0
\]
is the polynomial encoding of the scalar
$c:=|\varepsilon_k|+x_k$, which by the polar identity
$2(\operatorname{Re}\sqrt{z})^2=|z|+\operatorname{Re}(z)$
parametrises the level curve
$\operatorname{Re}\sqrt{z}=\operatorname{Re}\sqrt{\varepsilon_k(f_H)}$
in the complex $\varepsilon$-plane.
The square-root definition of $c$ is non-polynomial and therefore
unsuitable for quantifier elimination; the constraint $C=0$ replaces
it with an equivalent algebraic relation.

The \emph{separation polynomial}
\[
P_B(t,c,f_H):=N_y(t)^2-c^2Q_B(t)^2+2c\,N_x(t)Q_B(t)
\]
is obtained by clearing the denominator $Q_B(t)^2>0$ from the
parabolic inequality $y_b(t)^2-c^2+2c\,x_b(t)\ge0$, which is itself
the algebraic form of the condition
$\operatorname{Re}\sqrt{\varepsilon_b(t)}\ge \operatorname{Re}\sqrt{\varepsilon_k(f_H)}$.
Hence $P_B\ge0$ signifies that the boundary point $\varepsilon_b(t)$
violates the separation, i.e.\ the isotropic admissible set
reaches the Wiener upper bound.
The certification succeeds if the existential formula
$\exists\,f_H,t,c:C=0\wedge P_B\ge0$
is decided \texttt{False} for both branches by quantifier
elimination, establishing that the separation is strict everywhere.

\vspace{10pt}\par\noindent
\hrule height 1pt
\vspace{4pt}
\noindent\textbf{Algorithm 1} Algebraic Certification of Boundary Separation
\vspace{4pt}
\hrule height 0.5pt
\vspace{4pt}
\begin{algorithmic}[1]
\setlength{\itemsep}{1.2pt}

\Statex {\footnotesize\itshape $\triangleright$~Phase~1: set up constituent permittivities}
\State Set material constants $n_L,\kappa_L,n_H,\kappa_H$ as exact rational numbers.
\State Compute $\varepsilon_i'=n_i^2-\kappa_i^2$ and $\varepsilon_i''=2n_i\kappa_i$
       \quad for $i\in\{L,H\}$.
\State Define $f_L=1-f_H$,\quad $\Delta=\varepsilon_H-\varepsilon_L$,\quad $\varepsilon_k=f_L\varepsilon_L+f_H\varepsilon_H$.
\State Set $x_k=\operatorname{Re}\varepsilon_k$,\quad $y_k=\operatorname{Im}\varepsilon_k$,\quad $A=f_Lf_H\Delta^2=A_r+iA_i$.

\Statex {}
\Statex {\footnotesize\itshape $\triangleright$~Phase~2: test each Bergman--Milton boundary branch}
\For{each branch $B\in\{LH,\,HL\}$}
    \Statex \quad{\footnotesize\itshape $\triangleright$ Arc:\; $\varepsilon_b(t)=\varepsilon_k - f_Lf_H\Delta^2/D_B(t)$}
    \State Define the affine denominator $D_B(t)=D_r(t)+iD_i(t)$, where:
    \Statex \qquad $D_{LH}(t)=3\bigl[t\,\varepsilon_L+(1-t)\varepsilon_H\bigr],
                   \quad t\in I_{LH}(f_H)=\bigl[\tfrac{f_H}{3},\,1-\tfrac{f_L}{3}\bigr],$
    \Statex \qquad $D_{HL}(t)=3\bigl[t\,\varepsilon_H+(1-t)\varepsilon_L\bigr],
                   \quad t\in I_{HL}(f_H)=\bigl[\tfrac{f_L}{3},\,1-\tfrac{f_H}{3}\bigr].$
    \State Set $Q_B(t)=D_r(t)^2+D_i(t)^2$.
    \State Compute the $\varepsilon_b(t)$ numerator polynomials:
    \Statex \qquad $N_x(t)=x_k Q_B(t)-\bigl[A_r D_r(t)+A_i D_i(t)\bigr],$
    \Statex \qquad $N_y(t)=y_k Q_B(t)-\bigl[A_i D_r(t)-A_r D_i(t)\bigr].$
    \State Construct the separation polynomial:
    \Statex \qquad $P_B(t,c,f_H)=N_y^2-c^2 Q_B^2+2cN_x Q_B$.
    \State Define the level-curve constraint:
    \Statex \qquad $C(c,f_H)\;:=\;c^2-2x_k c-y_k^2=0,\quad c>0$.
    \State \textbf{Decide} over $\mathbb{R}$ the existential formula:
    \Statex \qquad $\exists\,f_H,t,c:\;0<f_H<1,\;\; t\in I_B(f_H),$
    \Statex \qquad\qquad\quad $C(c,f_H)=0,\;\; P_B(t,c,f_H)\ge0.$
    \State Record the output of the quantifier-elimination routine.
    \State Compute resultant $R_B(t,f_H)=\operatorname{Res}_c\!\bigl(P_B,\,C\bigr)$.
    \State \textbf{Decide} whether $R_B(t,f_H)=0$ admits a real solution with $0<f_H<1$ and $t\in I_B(f_H)$.
\EndFor

\Statex {}
\Statex {\footnotesize\itshape $\triangleright$~Phase~3: conclude}
\State \textbf{Conclude} boundary separation if both decision problems return $\mathrm{False}$ for both branches $B\in\{LH,\,HL\}$.

\end{algorithmic}
\vspace{4pt}
\hrule height 1pt
\vspace{14pt}\par
\[
\,\,\,
\]

\end{document}